\documentclass[12pt,twoside]{article}                 
 \usepackage{mathptmx}      
\usepackage{amsmath,amsfonts,amssymb,amsthm,url}
\usepackage{graphicx,epsfig} 
\usepackage[active]{srcltx}
\newtheorem{thm}{Theorem}
\newtheorem{lemma}{Lemma}

\newtheorem{prop}{Proposition}

\newtheorem{remark}{Remark}[section]

\newcommand{\abs}[1]{ \left | #1 \right |}
\newcommand{\dist}{{\, \rm dist}}
\newcommand{\norm}[1]{ \left \| #1 \right \|}

\newcommand{\Z}{\mathbb Z}
\newcommand{\ba}{\begin{eqnarray}}
\newcommand{\ea}{\end{eqnarray}}
\newcommand{\be}{\begin{equation}}
\newcommand{\ee}{\end{equation}}
\newcommand{\bmu}{\begin{multline}}
\newcommand{\emu}{\end{multline}}

\newcommand{\T}{\mathbb T}

\newcommand{\e}{\mathrm e}

\newcommand{\1}{{\mathbf 1}}
\newcommand{\R}{\mathbb R}
\newcommand{\N}{\mathbb N}
\newcommand{\C}{\mathbb C}

\newcommand{\E}{\mathbb E}

\newcommand\supp{{\rm supp}\,}

 \DeclareMathOperator{\tr}{Tr}
\DeclareMathOperator{\diam}{diam}

\theoremstyle{definition}
\newtheorem{defn}{Definition}
\begin{document}

\title{The weak localization for the alloy-type Anderson model on a cubic
lattice
}

\author{Zhenwei Cao\textsuperscript{1} and Alexander Elgart\textsuperscript{2}}
\footnotetext[1]{Department
of Mathematics, Virginia Tech., Blacksburg, VA, 24061, USA.  E-mail: zhenwei@vt.edu.  Supported in part by NSF
grant DMS--0907165.}
\footnotetext[2]{Department
of Mathematics, Virginia Tech., Blacksburg, VA, 24061, USA.  E-mail: aelgart@vt.edu.  Supported in part by NSF
grant DMS--0907165.}


\date{}

\maketitle

\begin{abstract}
We consider alloy type random Schr\"odinger operators on a cubic
lattice whose randomness is generated by the sign-indefinite
single-site potential. We derive Anderson localization for this
class of models in the Lifshitz tails regime, {\it i.e.} when the
coupling parameter $\lambda$ is small, for the energies $E\le
-C\lambda^{2}$.
\end{abstract}

\section{Introduction and main results}\label{sec:intro}
\subsection{Introduction}
The prototypical model for the study of localization properties of quantum
states of single electrons in disordered solids is the Anderson Hamiltonian $H_A$
on the lattice $\Z^d$. It consists of the sum of the finite difference Laplacian that describes the perfect crystal 
and a multiplication operator by a sequence of independent identically distributed
random variables that emulates the effect of disorder. The basic phenomenon, named Anderson localization after the physicist P. W.
Anderson, is that disorder can cause localization of electron states, which manifests
itself in time evolution (non-spreading of wave packets), (vanishing of) conductivity
in response to electric field, Hall currents in the presence of both magnetic
and electric field, and statistics of the spacing between nearby energy levels. 
The first property implies spectral localization, i.e. the spectral measure of $H_A$ is
almost surely pure point in some energy interval, and almost sure exponential decay of eigenfunctions there.

In this paper we study the properties of the more general class of systems, called the random alloy type
model,  in the three-dimensional setting. The action of the corresponding Hamiltonian $H^\lambda_\omega$
on the vector $\psi\in l^2(\Z^3)$ is described in equation below:
\begin{equation}\label{eq:model}
(H^\lambda_\omega \psi)(x) \ :=\ -\frac{1}{2}\,(\Delta
\psi)(x)+\lambda V_\omega(x) \psi(x)\,.
\end{equation}
Here $\Delta$ denotes the discrete Laplace operator,
\be\label{eq:Laplop}(\Delta \psi)(x) \ = \  \sum_{e\in \Z^3,\ |e|=1}\psi(x+e)\ - \ 6
\psi(x) \,,\ee
and $V_\omega$ stands for a random multiplication operator of the
form
\begin{equation}\label{eq:potentia}
V_\omega(x) \ = \ \sum_{i\in k\Z^3} \omega_i\, u(x-i)\,.
\end{equation}
Here $k\Z^3$ denotes the set of all points $i$ on $\Z^3$ which
are of the form \[i\ = \ (k_1z_1,k_2z_2,k_3z_3)\] for an arbitrary vector $z=(z_1,z_2,z_3)\in\Z^3$ 
and a given vector $k=(k_1,k_2,k_3)\in\Z^3$. The parameter $\lambda$ conveniently describes the strength of disorder. We note that the standard Anderson Hamiltonian $H_A$ corresponds to the choices $u(x)=\delta_{x}$; $k=(1,1,1)$ (so that $k\Z^3=\Z^3$). Here $ \delta_{x}$ stands for for the Kronecker delta function.

The localization properties are known to hold for $H_A$ in each of the following cases: 1)
high disorder ($\lambda\gg1$), 2) extreme energies, 3) weak disorder
away from the spectrum of the unperturbed operator, and 4) one dimension. Most of the mathematical
results on localization for operators with random potential in
dimensions $d>1$ have been derived using the multi-scale analysis (MSA) initiated by
Fr\"ohlich and Spencer \cite{FS} and by the fractional moment method (FMM)
of Aizenman and Molchanov \cite{AM}.  The Anderson localization problem for $H_A$ has been studied comprehensively, see for example \cite{Kirsch} and
references therein.

For a general single site function $u$ the situation is more complicated. If all the coefficients
$u(x)$ have the same sign,  the dependence of the spectrum
on $V$ is monotone - property that obviously holds true for $H_A$. Localization in such systems (in all regimes above) is relatively well understood by now, using the methods developed for the original Anderson model.
There is however  no physically compelling reason for a random tight binding alloy
model to be monotone, and the natural question is whether Anderson localization
still holds in the aforementioned regimes in the non-monotone case, i.e. when $u(x)$ are not all of the same sign. Mathematically,
the problem becomes especially acute when $\bar u :=\sum_x u(x) = 0$, see the discussion at the end of this section. In the strong disorder regime the recent preprint \cite{ess} extended FMM technique for a class of matrix Hamiltonians that include $H^\lambda_\omega$.

In this paper we study the localization properties of the alloy type models in the so called Lifshitz tails regime, namely localization of the states with energies that lie below the spectrum of the free Laplacian, for $\lambda\ll1$. The occurrence of localization for $H_A$ at energies near the band edges at
weak disorder is related to the rarefaction of low eigenvalues, and
was already discussed in the physical literature by I.~M. Lifshitz
in 1964, see Section $3$ in \cite{Lifshitz1},  and \cite{Lifshitz2}.

As far as the rigorous results are concerned, there is an extensive literature devoted to the proof of Lifshitz tails as well as Anderson localization in this regime for the original Anderson model $H_A$. In this case $\inf\sigma(H_A)= C\lambda$ almost surely, where the constant $C$ is smaller than zero provided that the random variables $\omega_x$ are i.i.d. and assume negative values with non zero probability (we will specify the assumptions on the randomness later on). One then is interested in showing that there is a non trivial interval $I$ of localization at the bottom of the spectrum of $H_A$, with $I=[C\lambda,E_0(\lambda)]$.  Among results in this direction, let us mention the work of M.~Aizenman \cite{A} that established localization for  $E_0(\lambda)=C\lambda+O(\lambda^{5/4})$. This was later improved by W-M.~Wang \cite{Wang} to $E_0(\lambda)=C\lambda+O(\lambda)$ and then by F.~Klopp \cite{Klopp2}  to $E_0(\lambda)=\tilde C\lambda^{7/6}$, with $\tilde C<0$. Motivated by the unpublished note
of T.~Spencer \cite{Spencer}, one of the authors \cite{elgart} derived the localization for the interval $I$ characterized by 
$E_0(\lambda)=\hat C\lambda^2$, with $\hat C<0$ (it is expected on physical grounds that beyond the spectrum of $H_A$ consists of delocalized states slightly above the threshold established in the latter paper). The proof utilizes the diagrammatic technique also employed in the current work. In fact a large portion of the effort spend here involves reducing the problem to the one studied earlier in \cite{elgart}, and generalizing the techniques presented there.

Less is known in the case of a general single site potential. When the average $\bar u$ of the single site potential is not equal to zero, $\inf\sigma(H_\omega^\lambda)\le C\lambda$ almost surely, with $C<0$ (see Section 5.1 of \cite{Klopp1}).  F.  Klopp  established the region of localization with $E_0(\lambda)=\tilde C \lambda^{7/6}$ in this case, \cite{Klopp2}. We show here that it persists at the
energy range $E\le \hat C\lambda^2$ regardless of the value of $\bar u$. However, for $\bar u=0$ the spectrum starts at $C \lambda^2$, and in order for the result to be non trivial in this situation, we have to show that $C<\hat C$.  At the end of this section we will describe the dipole model for which such condition can be proven.

\subsection{Assumptions}
\subsubsection{Randomness}
\begin{enumerate}
\item[($\mathcal A$)] The values of the random potential $\{\omega_i\}$ are
i.i.d. variables, with even, compactly
supported on an interval $J$, and bounded probability density
$\rho$. We will further assume (without loss of generality) that $J$ is centered around the origin and that the second moment satisfies $\E\, \omega_i^2=1$. The function $\rho$ is $\alpha$-H\"older continuous:
\be\label{eq:condondens} |\rho(x)-\rho(y)|\ \le \
K|x-y|^{\alpha}\,\max(\1_{J}(x),\1_{J}(y)\,),\ee
with $\alpha>0$ and where $\1_I$ stands for a characteristic function of the set $I$.
\end{enumerate}
\subsubsection{Single site potential}
We will focus our
attention on two somewhat extreme cases:
\begin{enumerate}
\item[$(\mathcal O)$]{\bf {Overlapping setup:}}
 The vector $k$ is of the form $k=(1,1,1)$, so that $k\Z^3=\Z^3$. This case corresponds in some
sense to the maximal random setting, and here we will impose rather
mild conditions on the {\it single site potential} function $u$.
Namely, we will assume that in this case $u(x)$ decays exponentially fast:
\be\label{eq:condonu} |u(x)| \ \le \ C e^{-A|x|}\,.\ee
\item[$(\mathcal N)$]{{\bf Non overlapping setup:}}
The vector $k$ is such that
$ \{\Theta-i\}\cap\Theta=\emptyset$ for all $0\neq i\in
k\Z^3$, and  $\supp\, u =: \Theta$ is compact.  This setting correspond to non overlapping  random potential. We will denote the corresponding primitive cell $\hat \Theta$, i.e. $\Theta\subset \hat \Theta$ and the translates of $\hat \Theta$ by $ i\in
k\Z^3$ tile $\Z^3$.
\end{enumerate}
\subsection{Notation and quantities of interest}

Let $e(p)$ denote the dispersion law, associated with the Fourier
transform of the Laplacian, $({\mathcal F} \Delta f)(p)= -2 e(p)\hat
f(p)$, where
$$\hat f( p):=({\mathcal F}f)( p):=\sum_{n\in\Z^3}e^{-i2\pi p\cdot
 n}f(n)\,,\quad   p \in\T^3:=[-1/2,1/2]^3\,,$$ with its inverse
$$\check g(  n )=\int_{\T^3}d^3 p \,e^{i2\pi  p \cdot
 n }f(  p )\,.$$
 One then computes
\begin{equation}\label{eq:e(p)}
e( p)=2\sum_{\alpha=1}^3\sin^2(\pi p\cdot e_\alpha)\,,
\end{equation}
where  $e_\alpha$ is a unit vector in the $\alpha$ direction. The
spectrum of the unperturbed operator $H^0_\omega$ is absolutely
continuous and consists of the interval $[0,6]$.

In what follows we will denote by $A(x,y)$ the kernel of the linear
operator $A$ acting on $l^2(\Z^3)$ (that is $A(x,y)=(\delta_y,
A\delta_x)=\langle y|A|x\rangle$, where $\delta_x$ is an indicator function of the site
$x\in \Z^3$, and $(\cdot,\cdot)$ denotes the inner product of
$l^2(\Z^3)$). We will use the concise notation $\int$ in place of
$\int_{(\T^3)^{k}}$ whenever it is clear from the context.

The  paper is devoted to the investigation of the
 properties of $H_\omega^\lambda$ for a typical
configuration $\omega$ in a weak disorder regime, namely at the
energy range $E<E_0$, where
\begin{enumerate}
\item For  Case $\mathcal O$,
\be\label{eq:E_0}E_0\ = \ E_0^{\mathcal O}\ = \  -2\lambda^2 \|\hat u\|^2_\infty-2\lambda^4\|\hat
u\|_\infty^4\,,\ee
for  $\lambda>0$ being sufficiently small.
\item For Case $\mathcal N$,
\be\label{eq:Esigm} E_0\ =\ E_0^{\mathcal N}\ =  \
-4n\,\lambda^2\,\|u\|^2_\infty\, \{(6-2E)^{\diam
\hat\Theta}\,|\hat\Theta|+1\}\,,\ee
where $\hat \Theta$ is a primitive cell described in Assumption $(\mathcal N)$, and  $|\hat \Theta|$ stands for the cardinality of the set $\hat\Theta$.
\end{enumerate}

\paragraph{Diagramatic expansion and self energy.}

The quantity of the most interest is the typical asymptotic behavior
of the so called Green function (also known as the two point
correlation function, the propagator)
\be\label{eq:R}R_{E+i\epsilon}(x,y) \ = \
(H^\lambda_\omega-E-i\epsilon)^{-1}(x,y)\ee
in the limit $\epsilon\searrow0$. It plays a crucial role in
determining, for instance, the conductivity properties of the
physical sample (whether it is an insulator or a conductor at a
given energy band). On a mathematical level, investigation of the
propagator can yield an insight on the typical spectrum of
$H_\omega^\lambda$ at the vicinity of $E$.
\par 
The technical assertions in this paper are proven using Feynman diagrammatic expansion for $R_{E+i\epsilon}$ around the unperturbed resolvent, i.e. the one with $\lambda=0$ (Section \ref{sec:repr}). The rate of convergence for this expansion in the limit $\epsilon\rightarrow0$ depends strongly on the value of $E$, and sets our limitations for the length of the interval $I$ where we can prove localization. One can eliminate certain terms in this expansion (the so called tapole contributions) that are especially problematic from the convergence point of view, by modifying the unperturbed Hamiltonian. The corresponding addend $\sigma$ is called the {\it self energy} of the model.
\par 
In the case $(\mathcal O)$ we
define the self energy  by the solution of the
self-consistent equation
\begin{equation}\label{eq:self}
\sigma(p,E+i\epsilon)\ =\ \lambda^2\int_{\T^3}\,d^3q\,
\frac{\left|\hat
u(p-q)\right|^2}{e(q)-E-i\epsilon-\sigma(q,E+i\epsilon)}\,.
\end{equation}
The relevant properties of the solution of (\ref{eq:self}) are
collected in Appendix \ref{sec:appendI}. In particular, it has a
single valued solution $\sigma(p,E+i\epsilon)$ for all $\epsilon$,
all $p\in \T^3$ and all values of $E$ that meet the condition
\eqref{eq:E_0}. The function $\sigma$ satisfies
\[\|\sigma\|_\infty\ \le\ \min(-E-2\lambda^4\|\hat
u\|_\infty^4,2\lambda^2 \|\hat u\|^2_\infty).\]
\par

In the case $(\mathcal N)$, the introduction of the self energy term requires some additional preparation. Let $n=|\hat \Theta|$. Let $\{x_i\}_{i=1}^n$ be some enumeration of the
sites of the primitive cell $\hat \Theta$. Let $D$ be the diagonal
$n\times n$ matrix with $D_{ii}=u(x_i)$, and let $\mathcal D$ be its periodic extension to $\ell_2(\Z^3)$, i.e.
\be\label{eq:mathcalD}{\mathcal D}(x,y) \ = \ u(x\hspace{-.2 cm}
\mod \hat\Theta)\,\delta_{x-y}\,.\ee
 For any $n\times n$
matrix $\sigma$,  we construct the periodic operator $\Sigma$ acting on  $\Z^3$ by defining its kernel as
\be\label{eq:Sigm}\Sigma(x,y) \ := \ \begin{cases}
\sigma_{ij}  & x,y\in\hat \Theta+l \hbox{ for some } l\in k\Z^3;\  x\hspace{-.2 cm} \mod \hat\Theta=x_i, \ y \hspace{-.2 cm}\mod
\hat\Theta=x_j\cr
0 &  \hbox{ otherwise } \cr
\end{cases}\,.\ee
We can now define an $n\times n$ matrix $S$   given by
\be \label{eq:Selem} S_{ij} \ = \ \langle
x_i\left|\left(-\Delta/2-E-i\epsilon-\Sigma\right)^{-1} \right|x_j\rangle\,;\quad x_i\,,x_j\in\hat \Theta\,.
\ee
 Then the self energy term $\sigma$ in the case $(\mathcal N)$ is
going to be the $n\times n$ matrix, which satisfies
\begin{equation}\label{eq:self'}
\sigma\ =\ \lambda^2\,D\,
S \, D \,.
\end{equation}
The solution of \eqref{eq:self'}  enjoys properties similar to the
one of \eqref{eq:self}, namely it is unique for all $E<E_0$ (where
$E_0$ is given by \eqref{eq:Esigm}) and
$|\epsilon|<2n\,\lambda^2\,\|u\|^2_\infty$, and satisfies
\be\label{eq:propsigm}
\|\sigma\| \ \le \ 2n\,\lambda^2\,\|u\|^2_\infty  \,.
\ee
We defer further  discussion of the properties of $\sigma$ to Appendix \ref{sec:appendI}.

\subsection{Main result}\label{subs:res}
The hallmark of localization is a rapid decay of the Green function at energies
in the spectrum of $H_\omega$, for the typical configuration
$\omega$. This behavior can be linked to the
non-spreading of wave packets supported in the corresponding energy
regimes and various other manifestations of localization. Our main result, Theorem \ref{thm:main} below, establishes this behavior of the Green function at the band edges of the spectrum, by comparing it with the asymptotics of the free Green
function.
%
\begin{thm}[Anderson Localization for Lifshitz tails regime]\label{thm:main}
For $H_\omega^\lambda$ as above that satisfies Assumption ($\mathcal A$) and either ($\mathcal O$) or ($\mathcal N)$, for any $\nu>0$ there
exists $\lambda_0(\nu)$ such that for all
$\lambda<\lambda_0(\nu)$  the
spectrum of $H_{\omega}$ within the set $E\le E_0-\lambda^{4-\nu}$ is almost-surely of the pure-point type, and the corresponding
eigenfunctions are exponentially localized.
\end{thm}
%
\subsection{Discussion}
It should be noted that our method works most effectively  when the average $\bar u$ of the single site potential is not equal to zero. In this case $\inf\sigma(H_\omega^\lambda)\le C\lambda$ almost surely, with $C<0$ (see Section 5.1 of \cite{Klopp1}). Hence for $\lambda$ sufficiently small Theorem \ref{thm:main} establishes the localization at the bottom of the spectrum of  $H_\omega^\lambda$. We are only aware of one result on the Anderson localization  in the regime discussed here  for a non sign definite single site potential:  F.  Klopp  proved the weak disorder localization for $E<-\lambda^{7/6}$  in three dimensions, \cite{Klopp2}. Since for $\bar u=0$ case $-\inf\sigma(H_\omega^\lambda)= O(\lambda^2)$, his result does not provide an answer on whether  the bottom of the spectrum is localized or not.

When  the average of the single site potential vanishes, we expect our method to yield the non trivial result when the minimizing configurations of the random potential look flat. That is essentially a reason why  we can cover say the dipole potential in the $(\mathcal O)$ case below, the observation we owe to G\"unter Stolz. In this case the expansion around the free Green function is a sensible procedure to do. However, in the case $(\mathcal N)$ the sufficient symmetry of the single site potential can cause the minimizing configuration to become periodic. For reflection symmetric $u$ this was shown in \cite{Stolz} (for the continuum analogue of $H_\omega^\lambda$). In this situation, the free Green function does not capture the right features of the problem, and as a result our method fails to achieve the non trivial result in this case. It is worth noticing that exploiting the specific knowledge about the minimizing potential, one can show Lifshitz tails and consequently Anderson localization for such reflection symmetric single site potential \cite{StolzKlopp}.

In this paper we consider the cubic lattice, that is $d=3$ case. Similar (in fact better) results can be established for a higher dimensional case, but not for $d<3$. Mathematically it is related to the nature of the point singularity of the propagator  $e(p)$ at zero energy - it is integrable  for $d\ge3$. This fact allows us to control the underlying Feynman series.

We end the discussion with an example pertaining to the case $\hat u(0)=0$ where the technique developed in this paper allows to get  a meaningful result by improving the bound on the threshold energy $E_0$:

Consider the single site potential $u_d$ of the dipole type, i.e.
\be\label{eq:defdip}
u_d(x) \ = \
\begin{cases}
1  & x=0  \cr
-1  & x= e_1  \cr
0 &  \hbox{ otherwise } \cr
\end{cases}\,.
\ee
%
\begin{prop}\label{prop:dipole}
For the single site potential $u_d$ we have
\[\inf\sigma(H^\lambda_\omega)\ <\ -2\lambda^2+O(\lambda^3)\]
 almost surely. The statement of Theorem \ref{thm:main} holds true for all energies $E$ that satisfy
\[E<E_d\ :=\ -(1+\lambda)\lambda^2\,.\]
\end{prop}

\section{Outline of the proof}
We will establish Theorem \ref{thm:main} using the multiscale
analysis (MSA) method. It requires two inputs: The initial volume
 and Wegner estimates.
\subsection{Diagrammatic expansion}
The first ingredient in the proof of Theorem \ref{thm:main} can be
established in the framework of Feynman graphs perturbation
technique.  Throughout the text various Green functions will appear. We will denote by $G_E(x,y)$ the free Green function, i.e.
\be\label{eq:G(xy)} G_E(x,y)\ := \ \langle x|\,\left(-\Delta/2-E\right)^{-1}\,|y\rangle \,.\ee
We characterize its relevant properties in Appendix \ref{sec:append}. It will be used in some of the proofs as comparison with the full Green function $R_{E+i\epsilon}(x,y)$, defined in \eqref{eq:R}. Whenever it is clear from the context, we will suppress the energy dependence of $R_{E+i\epsilon}$, and just use $R$ (respectively $R(x,y)$) for the full resolvent (the full Green function).

The following representations for a Green function
$R(x,y)$ provide the key technical tool for us:
\begin{lemma}[Decomposition of $R(x,y)$ in ($\mathcal N$) case]\label{lemma:rep_matr}
Let $E_{\mathcal N}^*=-E+E_0/2$, and let
\be\label{eq:deltaE'}\delta_{\mathcal N} \ := \ \sqrt{E_{\mathcal N}^*/3}\,.\ee
 Then for any integer $N$ and energies $E<E_0$ with $E_0$
satisfying \eqref{eq:Esigm} we have the decomposition
\begin{equation}\label{eq:mdec_matr}
R(x,y)\ = \ \sum_{n=0}^{N-1}A_n(x,y) \ + \ \sum_{z\in Z^3}\tilde
A_N(x,z)R(z,y)\,,
\end{equation}
with $A_0(x,y)=R_r(x,y)$ (the latter quantity is defined in \eqref{eq:Rsigma} below), and where the kernels
$A_n,\,\tilde A_N$ satisfy bounds
\begin{eqnarray}
\E\, |A_n(x,y)|^2& \le & (4n)!\,  E_{\mathcal N}^*\,\left(C(E_{\mathcal N}^*)\
\frac{\lambda^{2}}{\sqrt{E_{\mathcal N}^*}}\right)^n\,e^{-\delta_{\mathcal N}\,|x-y|}\,,\
 \ n\ge1 \,;\label{eq:l_matr}
\\
\E\, |\tilde A_N(x,y)|& \le & \,\sqrt{(4N)!}\,\left(C(E_{\mathcal N}^*)\
\frac{\lambda^{2}}{\sqrt{E_{\mathcal N}^*}}\right)^{N/2}\,e^{-\delta_{\mathcal N}|x-y|/2}\,,\
 \ N>1 \,;\label{eq:lt_matr}
\end{eqnarray}  where $C(E_{\mathcal N}^*)=K\ |\hat \Theta|\|{\mathcal D}\|\ln^{9}(E_{\mathcal N}^*)$ and $K$ is  some generic constant.

The zero order contribution $A_0$ satisfies
\begin{multline}\label{eq:stand_matr}
|A_0(x,y)|\ \le \
G_{E_{\mathcal N}^*}(x,y) \ \le \ C\,e^{-|x-y|\frac{\delta_{\mathcal N}}{3\sqrt 3}} \,
\max\left(\sqrt{E_{\mathcal N}^*}\,,\,(1+|x-y|)^{-1}\right)
\end{multline}
for all $x,y\in\Z^3$.
\end{lemma}

\par

We now formulate the parallel result for the case (${\mathcal O}$). To this end, we introduce some additional notation first.
For the parameter $E_{\mathcal O}^*$ that satisfies
$-E+E_0^{\mathcal O}>E_{\mathcal O}^*>0$ with $E_0^{\mathcal O}$
defined in \eqref{eq:E_0}, we set
\be\label{eq:deltaE}\delta_{\mathcal O} \ := \ \frac{ \sqrt{E_0^{\mathcal O}-E-E_{\mathcal O}^*}}{ \sqrt 6 \pi}\,.\ee
%
\begin{lemma}[Decomposition of $R(x,y)$ in ($\mathcal O$) case]\label{lemma:rep}
 For any integer $N$ and energies $E<E_0$ with $E_0$
satisfying \eqref{eq:E_0} we have the decomposition
\begin{equation}\label{eq:mdec}
R(x,y)\ = \ \sum_{n=0}^{N-1}A_n(x,y) \ + \ \sum_{z\in Z^3}\tilde
A_N(x,z)R(z,y)\,,
\end{equation}
with $A_0(x,y)=R_r(x,y)$ (the latter kernel is defined in \eqref{eq:R_r_O} below), and where the (real valued) kernels
$A_n,\,\tilde A_N$ satisfy bounds
\begin{eqnarray}
\E\, |A_n(x,y)|^2\ \le \ (4n)!\, E_{\mathcal O}^*\,\left(C(E_{\mathcal O}^*)\
\frac{\lambda^{2}}{\sqrt{E_{\mathcal O}^*}}\right)^n\,e^{-\delta_{\mathcal O}\,|x-y|}\,,\
 \ n\ge1 \,;\label{eq:l}
\\
\E\, |\tilde A_N(x,y)|\ \le \ \sqrt{(4N)!}\,\left(C(E_{\mathcal O}^*)\
\frac{\lambda^{2}}{\sqrt{E_{\mathcal O}^*}}\right)^{N/2}\,e^{-\delta_{\mathcal O}|x-y|/2}\,,\
 \ N>1 \,;\label{eq:lt}
\end{eqnarray}  where
\[C(E_{\mathcal O}^*)\ = \ K \left(\|\hat u\|_{\infty}\,+\, A^{-4} {\delta_{\mathcal O}}\right)\,\ln^{9}(E_{\mathcal O}^*)\]
for some generic constant $K$ and $A$ being a parameter introduced in \eqref{eq:condonu}.

The zero order contribution $A_0$ satisfies
\begin{equation}\label{eq:stand}
 |A_0(x,y)|\ \le \ 2 \,e^{-\frac{\delta_{\mathcal O}}{3\sqrt 3}|x-y|}
\end{equation}
for all $x,y\in\Z^3$.
\end{lemma}
\begin{remark}
There is a certain balance between the parameters $E^*_{\mathcal O}$ and $\delta_{\mathcal O}$ in the above assertion. Namely, increasing the former improves on the prefactor in front of the exponential decay in \eqref{eq:l}, but since it decreases the latter, the control over the rate of the exponential decay becomes poorer.
\end{remark}
\begin{remark} The representations \eqref{eq:mdec_matr} and \eqref{eq:mdec} are  resolvent type expansions (see Lemma \ref{lem:de} below for details). If one applies the rough norm bound
on each factor of the resolvent there, the denominator in
\eqref{eq:l_matr} and \eqref{eq:l} will contain $E^*$ rather than its square root. The
improvement is achieved using the Feynman diagramatic technique
(Section \ref{sec:repr}).
\end{remark}
\par
One then looks for the optimal value $N$ to stop the corresponding expansion -
note that the increasing  factor of $(4N)!$ in $A_N(x,y)$ competes
with the decreasing factor $(\lambda^4 E^*)^{N/2}$.

The choice $E^*=\lambda^{4-\nu}/2$ has the effect that
\begin{equation}\label{eq:refe}C(E^*)\ \frac{\lambda^{2}}{\sqrt{E^*}} \ \le \
\lambda^{B\nu}\,,\quad 0 < B < 1\,,\end{equation}
which
suffices to control \eqref{eq:l_matr} -- \eqref{eq:lt_matr} (respectively \eqref{eq:l} -- \eqref{eq:lt}). We note that in the range of energies $E<E_0^{\mathcal O}-\lambda^{4-\nu}$, the above choice for $E^*_{\mathcal O}$ implies $\delta>\lambda^{4-\nu}/(2\pi)^2$. It turns out that
the appropriate choice for N should satisfy
\[(4N)!\,\left(\frac{C(E^*)\lambda^2} {\sqrt{E^*}}\right)^{N}\approx
e^{-4N}\] (see the next section for details). In terms of the
$\lambda$ - dependence, it corresponds to
$N\sim\lambda^{-b\nu}$ for $b < B$.
\subsection{Wegner estimate}
The initial volume estimate that enters into MSA requires us to get
rid of the imaginary part $\epsilon$ of the energy, present in the
formulation of Lemma \ref{lemma:rep}. To this end, we will use the
Wegner estimate below, that itself is an important ingredient of
MSA. It will be established using the idea of F. Klopp \cite{Klopp}.
While it has the correct (linear)  dependence on the length of the interval $I$, 
the dependence on the volume of the estimate below is not
optimal. The optimal, linear dependence on $\Lambda$ for the
continuum models with absolute continuous density $\rho$ was developed in \cite{HK}, but the trade-in is that the $I$-dependence in the Wegner estimate thereof is worse than ours. The derivation below has an advantage of being  completely
elementary and the estimate itself is sufficient for our purposes.
\begin{thm}[Wegner estimate for H\"older continuous densities]\label{lem:Wegner}
Let $I$ be an open interval of energies such that
\[D_I\ :=\ \dist(I,\sigma(-\Delta/2))\ >\ 0\,.\]
Then  for either $\mathcal O$ or $\mathcal N$ cases we have
\be \label{eq:wegn} \E\,\tr P_I(H_\omega^{\Lambda,\lambda}) \ \le  \
C\, |I|\, |\Lambda|^{\frac{1+\alpha}{\alpha}}\,(D_I)^{-1}\,, \ee
where  $H_\omega^{\Lambda,\lambda}$ denotes a natural restriction of $H_\omega^{\lambda}$ to $\Lambda\subset\Z^3$,  the constant $\alpha$ is defined in ($\mathcal A$), and the constant $C$ depends on $J$ and $K$.
\end{thm}
\begin{remark} The dependence on volume $\Lambda$ blows up as $\alpha\searrow 0$.
Thus the technique doesn't work when $\rho$ is concentrated in finitely many points. See
Bourgain and Kenig \cite{BK} for the Bernoulli alloy type model, where $\rho$ is concentrated
on $\{0,1\}$, in continuous settings.
\end{remark} 
The rest of the paper is organized as follows:  We derive the main results based on the technical statements above in Section \ref{sec:proofs}. We then describe the procedure which allows us to have a tighter control over the resolvent expansion in Section \ref{sec:renorm}. We prove Lemmas \ref{lemma:rep_matr} and \ref{lemma:rep} in Section \ref{sec:repr}.  Auxillarly technical statements
 are collected in Appendices \ref{sec:append} and \ref{sec:appendI}.

\section{Proofs of the main results}\label{sec:proofs}
\begin{proof}[Proof of Theorem \ref{thm:main}] 
\hspace{7cm} \vskip 3mm \noindent
The proofs of the auxiliary statements, namely Theorem \ref{lem:Wegner} (accordingly Lemmas \ref{lemma:rep_matr}, \ref{lemma:rep}), will be postponed until later in this section  (until Section \ref{sec:repr}). 
\vskip 3mm
\noindent We prove the assertion simultaneously
for ($\mathcal N$) and ($\mathcal O$) cases, so we are going to drop
subscripts $\mathcal N$ or $\mathcal O$ until the very end of the
proof.  Let us denote by $H_\omega^{\Lambda,\lambda}$ the natural
restriction of $H_\omega^{\lambda}$ to $\Lambda\subset\Z^3$, namely,
$H_\omega^{\Lambda,\lambda}(i,j)=H_\omega^\lambda(i,j)$ if $(i,j)\in\Lambda
\times\Lambda$ and $H_\omega^{\Lambda,\lambda}(i,j)=0$ otherwise. Let
$\Lambda^c:=\Z^3\setminus\Lambda$, and let $\partial\Lambda=\{i\,|\,\exists j\,\rm{s.t.}\,
(i,j)\in\Lambda\times\Lambda^c,\,\rm{dist}(i,j)=1\}$ be the
boundary of the set $\Lambda$.  We define the decoupled Hamiltonian
$H_\Lambda$ to be
\[H_\Lambda=H_\omega^{\Lambda,\lambda}\oplus H_\omega^{\Lambda^c,\lambda}\,,\]
and will denote by $R_\Lambda(E)$ the corresponding resolvent (i.e.
$R_\Lambda(E)=(H_{\Lambda}-E)^{-1}$). For $L >0$ and $x \in \Z^d$ we
denote by $\Lambda_{L,x} = \{y \in \Z^d : \lvert x-y \rvert_\infty
\leq L\}$ the cube of side length $2L$. Our first objective is to
derive the bound for $\E \lvert
R_{\Lambda_{L,x}}(E+i\epsilon;x,w)\rvert$, for $w\in
\partial\Lambda$. To this end, we observe
\begin{multline}\label{eq:ini}
\E\lvert R_{\Lambda_{L,x}} (E+i\epsilon;x,w) \rvert \\ \le \
\E\lvert R(E+i\epsilon;x,w) \rvert \,+\, \E\lvert R_{\Lambda_{L,x}}
(E+i\epsilon;x,w) \,-\,R(E+i\epsilon;x,w)\rvert\\ = \ \E\lvert
R(E+i\epsilon;x,w) \rvert \,+\, \E\lvert \left(R
(H_\omega-H_{\Lambda_{L,x}})\,R_{\Lambda_{L,x}}\right)(x,w)\rvert
\\ \le \
\E\lvert R(E+i\epsilon;x,w) \rvert \\ + \ \E
\sum_{k\in\partial\Lambda_{L,x}^c}\,\lvert
R(E+i\epsilon;x,k)\rvert\,\lvert \left(R_{\Lambda_{L,x}}
(H_\omega-H_{\Lambda_{L,x}})\, \right)(k,w)\rvert\,.
\end{multline}
We can estimate
\begin{equation}\label{eq:diff}
\sum_{k\in\partial\Lambda_{L,x}^c}\lvert \left(R_{\Lambda_{L,x}}
(H_\omega-H_{\Lambda_{L,x}})\, \right)(k,w)\rvert \ \le \
\frac{C_0}{\epsilon}\, \lvert\partial\Lambda\rvert \ = \
C\,\frac{L^2}{\epsilon}\,,
\end{equation}
hence
\begin{equation}\label{eq:boxgr}
\E\lvert R_{\Lambda_{L,x}} (E+i\epsilon;x,w) \rvert \ \le\
C\,\frac{L^2}{\epsilon}\, \max_{k\in\partial\Lambda_{L,x}^c}\,
\E\lvert R(E+i\epsilon;x,k)\rvert\ \,.
\end{equation}
On the other hand, for any $k\in\partial\Lambda_{L,x}^c$, we have
${\rm dist}(x,k)=L+1$, and Lemma \ref{lemma:rep_matr}
(respectively Lemma \ref{lemma:rep}) ensures that
\begin{multline}\label{eq:glob}
\E\lvert R(E+i\epsilon;x,k)\rvert \ \le \ \sum_{n=0}^{N-1}\E\lvert
A_n(x,k)\rvert \, + \, \sum_{z\in Z^3}\E\lvert \tilde
A_N(x,z)R(E+i\epsilon;z,k)\rvert\\
\le \ \sum_{n=0}^{N-1}\Big\{\E\lvert A^2_n(x,k)\rvert \Big\}^{1/2}\,+\,\frac{1}{\epsilon}\,
\sum_{z\in Z^3}\E\lvert \tilde
A_N(x,z)\rvert\\ \le \ \sum_{n=0}^{N-1}\sqrt{(4n)!}\,
\sqrt{E^*}\,\left(C(E^*)\
\frac{\lambda^{2}}{\sqrt{E^*}}\right)^{n/2}\,e^{-\delta\,|x-k|/2}\\
+ \ \sum_{z\in Z^3}\frac{\sqrt{(4N)!}}{\epsilon}\, \left(C(E^*)\
\frac{\lambda^{2}}{\sqrt{E^*}}\right)^{N/2}\,e^{-\delta\,|x-z|/2}\\
\le \ e^{-\delta L/2}\sum_{n=0}^{N-1}\sqrt{E^*(4n)!}\, \left(C(E^*)\
\frac{\lambda^{2}}{\sqrt{E^*}}\right)^{n/2}\\ +\ C\,
\frac{\sqrt{(4N)!}}{\epsilon}\,\delta^{-1} \left(C(E^*)\
\frac{\lambda^{2}}{\sqrt{E^*}}\right)^{N/2}
\end{multline}
Choosing
\be\label{eq:choiceN}
(4N)^4\ =\ \frac{\sqrt{E^*}}{C(E^*)\,\lambda^2} \,,\ee
one obtains, using the  Stirling's approximation, that the summation
over the index $n$ is bounded by a constant and
\[(4N)!\left(\frac{C(E^*)\lambda^2} {\sqrt{E^*}}\right)^{N}\ \approx \
e^{-N}\,.\]
Hence, for such a value of $N$ we have
\begin{equation}\label{eq:pre}
\E\lvert R(E+i\epsilon;x,k)\rvert \ \le \  C\,\left(
e^{-\delta\,L/2}+ \frac{e^{-N}}{\epsilon \,\delta}\right)\,.
\end{equation}
Combining this bound with \eqref{eq:boxgr}, we obtain
\be\label{eq:estgrof} \E\lvert R_{\Lambda_{L,x}} (E+i\epsilon;x,w)
\rvert \ \le \ C\,\frac{L^2}{\epsilon}
\left[ e^{-\delta\,L/2}\,+\, \frac{e^{-N}}{\epsilon
\delta}\right]\,.\ee
Let $I=[E-\epsilon^{1/4},E+\epsilon^{1/4}]$, and let
\[G(I)\ :=\ \left\{\omega\in\Omega: \ \sigma(H_{\Lambda_{L,x}})\cap
I=\emptyset\right\}\,.\]
 For any $\omega\in G(I)$ we have by
the first resolvent identity
\[\left|R_{\Lambda_{L,x}} (E+i\epsilon;x,w)-R_{\Lambda_{L,x}}
(E;x,w)\right|\le \epsilon^{1/2}\,.\]
Pairing this bound with \eqref{eq:estgrof} and using Chebyshev's
inequality, we get that
\begin{multline}\label{eq:estgrof'}{\rm Prob} \left\{\omega\in G(I):\
\lvert R_{\Lambda_{L,x}} (E;x,w) \rvert \ge
C\frac{L^2}{\epsilon^{5/4}} \left[e^{-\delta\,L/2}\,+\,
\frac{e^{-N}}{\epsilon \delta}\right]\,+\,\epsilon^{1/4}\right\}\\
\le \ \epsilon^{1/4}\,.\end{multline}
The Wegner estimate \eqref{eq:wegn} implies that
\be \label{eq:wegner} {\rm Prob} \bigl\{
\sigma(H_{\Lambda_{L,x}})\cap I \neq\emptyset\bigr\} \ \leq \   C\,
|I|\, \,|\Lambda_{L,x}|^{\frac{\alpha+1}{\alpha}}\,(D_I)^{-1}\  = \ C
\,\epsilon^{1/4}\,(D_I)^{-1}\, L^{3\frac{\alpha+1}{\alpha}} \,.\ee
Combining \eqref{eq:estgrof'} and \eqref{eq:wegner} we arrive at
\begin{multline}\label{eq:estgrof''}{\rm Prob} \left\{\lvert R_{\Lambda_{L,x}}
(E;x,w) \rvert \ge C\frac{L^2}{\epsilon^{5/4}}
\left[e^{-\delta\,L/2}\,+\, \frac{e^{-N}}{\epsilon
\delta}\right]\,+\,\epsilon^{1/4}\right\}\\  \le \
C\epsilon^{1/4}\,(D_I)^{-1}\, L^{3\frac{\alpha+1}{\alpha}} \,.\end{multline}
We are now in position to set the values for the various parameters in the above formula, in terms of the single parameter $\lambda$. We first note that since $E\le E_0-\lambda^{4-\nu}$ with $0<\nu<1$ and $E_0$ defined in \eqref{eq:E_0} - \eqref{eq:Esigm}, it is allowed to choose $E^*=\lambda^{4-\nu}/2$. It then follows from \eqref{eq:refe} and \eqref{eq:choiceN} that $N\sim\lambda^{-B'\nu}$ with $0< B'<1/4$, for $\lambda$ small enough. Next, the parameter $\delta$ originating from Lemmas \ref{lemma:rep_matr} and \ref{lemma:rep} satisfies $\delta\sim\lambda^{2-\nu/2}$ for $\lambda$ small enough. We finally choose $L=\lambda^{-2}$ and $\epsilon=e^{-5\lambda^{-B'\nu/2}}$. Plugging it all into \eqref{eq:estgrof''} we obtain that for the small values of $\lambda$ the following initial volume estimate holds true:
\be\label{eq:estgrofa}{\rm Prob} \left\{\lvert R_{\Lambda_{\lambda^{-2},x}}
(E;x,w) \rvert \ge e^{-\lambda^{-B'\nu/2}}\right\}\  \le \
e^{-\lambda^{-B'\nu/2}}\,.\ee
The initial volume estimate \eqref{eq:estgrofa} together with the Wegner estimate \eqref{eq:wegner} provide the necessary input for MSA for small $\lambda$, and   the result follows from say Theorem 2.4 of \cite{GK}.

\end{proof}
\begin{proof}[Proof of Theorem \ref{lem:Wegner}]
\hspace{7cm} \vskip 3mm \noindent
In the sequel we  will use the fact that any $\alpha$-H\"older continuous non negative function $\rho$  admits a Lipschitz approximation by means of a   non negative function $\rho_K$ such that $\rho_K$ is $K$-Lipschitz supported on $J$, and 
\be\label{eq:Lapprox}
\|\rho-\rho_K\|_{\infty}\ \le \ C\,K^{\frac{-\alpha}{1-\alpha}}\,.\ee
It is then follows from $\|\rho\|_1=1$ that 
\be\label{eq:rhoK}
\|\rho_K\|_{1} \ \le \ 1\,+\,C\,\left(K^{\frac{-\alpha}{1-\alpha}}\right)\,|J|\,.\ee
Observe that for any random quantity $F_\omega$ that
depends on the random variables $\omega_i$, with $i\in\Lambda$, we
have for any $\delta$
\be \E F_\omega \ = \ \int F_\omega\prod_{i\in\Lambda}
\rho(\omega_i) d \omega_i \ = \ \frac{1}{\delta}\,
\int_{1}^{1+\delta} v^{|\Lambda|} dv \int
F_{v\hat\omega}\prod_{i\in\Lambda} \rho(v\hat \omega_i) d
\hat\omega_i\,.\ee
So, in order to evaluate $\E\ P_I(H_\omega^{\Lambda,\lambda})$,  
we can first integrate over the fictitious random variable $v$. We will rely on the following simple  statement.
\begin{lemma}\label{lem:aux_Wegner}
Let $A,B$ be  hermitian $n\times n$ matrices. Let $\{\lambda_k(A)\}_{k=1}^n$,  respectively  $\{\lambda_k(B)\}_{k=1}^n$ be the set of corresponding eigenvalues in the ascending order, and suppose that 
\[0<\alpha=\lambda_1(B)\le \lambda_n(B)=\beta\,.\]
Let $I,J$ be the intervals $[a,b]$ and $[c,d]\subset \R_+$ accordingly  and let $P_I$ denote the characteristic function of the interval $I$. Then we have 
\begin{equation}\label{eq:aux_Weg}
\int_J \tr P_I(A+x B) dx \ \le \ \alpha^{-1}\, |I| \,\tr P_{\hat I}(A)\,,
\end{equation}
where $\hat I=[a-\beta d,b-\alpha c]$.
\end{lemma}
For all $v\in[1,1+\delta]$   we can use \eqref{eq:Lapprox} to bound
\begin{eqnarray}
\rho(v\hat \omega_i)&\le &\rho_K(v\hat
\omega_i) \,+\, C\,K^{\frac{-\alpha}{1-\alpha}}\1_{J}(v\hat \omega_i)\nonumber\\ &\le & \rho_K(\hat
\omega_i) \,+\,\left(C\,K^{\frac{-\alpha}{1-\alpha}}\,+\,K\delta\,|J|\right)\,\1_{J'}(\hat \omega_i)
\label{eq:rhobound}
\\ &=:&f(\hat \omega_i)\,, \nonumber
\end{eqnarray}
where $J'=(1+\delta)J$. 

If $E$ is a middle point of $I$, we can write $P_I(H_{v\hat \omega}^{\Lambda,\lambda})=P_{I_0}(H_{v\hat \omega}^{\Lambda,\lambda}-E)$, where $I_0$ is centered at origin and has the same width as $I$. Let $B:=-\frac{1}{2}\Delta-E$, then $H_{v\hat \omega}^{\Lambda,\lambda}-E=B+v \lambda V_\omega^\Lambda$. We note that $B$ satisfies the same properties as its counterpart in Lemma \ref{lem:aux_Wegner} above, with $\alpha=D_I$ (from Lemma \ref{lem:Wegner}) and $\beta=6+D_I$. Since 
\[P_{I_0}(B+v \lambda V_\omega^\Lambda) \ = \ P_{v^{-1} I_0}(v^{-1}B+ \lambda V_\omega^\Lambda) \ \le \ P_{ I_0}(v^{-1}B+ \lambda V_\omega^\Lambda)\] 
for $v\in[1,1+\delta]$, we can estimate
\begin{eqnarray}\label{eq:C_l}
\E\ P_I(H_\omega^{\Lambda,\lambda}) & \le & \frac{1}{\delta}\,
\int
\prod_{i\in\Lambda} f(\hat \omega_i) d\hat\omega_i\ \int_{1}^{1+\delta} v^{|\Lambda|} P_{ I_0}(v^{-1}B+ \lambda V_\omega^\Lambda)dv \\ \nonumber 
 & \le & \frac{1}{\delta}\,
\left(1+\delta\right)^{|\Lambda|}\  \int
\prod_{i\in\Lambda} f(\hat \omega_i) d\hat\omega_i \ \int_{1}^{1+\delta}  P_{ I_0}(v^{-1}B+ \lambda V_\omega^\Lambda)dv \\ \nonumber 
& = & \frac{1}{\delta}\,\left(1+\delta\right)^{|\Lambda|}\  \int
\prod_{i\in\Lambda} f(\hat \omega_i) d\hat\omega_i \ \int_{1}^{1+\delta}   P_{ I_0}(xB+ \lambda V_\omega^\Lambda)x^{-2}dx\\ \nonumber 
& \le & \frac{1}{\delta}\,\left(1+\delta\right)^{2+|\Lambda|}\  \int
\prod_{i\in\Lambda} f(\hat \omega_i) d\hat\omega_i \ \int_{J_\delta} P_{ I_0}(xB+ \lambda V_\omega^\Lambda)dx\,,
\end{eqnarray}
where 
\[ J_\delta \ = \ \left[\left(1+\delta\right)^{-1},1\right]\,.\]
Applying Lemma \ref{lem:aux_Wegner} with the choice $A=\lambda V_\omega^\Lambda$, we get 
\be
\E\ P_I(H_\omega^{\Lambda,\lambda})  \ \le \ \frac{1}{\delta}\,\left(1+\delta\right)^{2+|\Lambda|}\,(D_I)^{-1}\, |I|\   \int
\prod_{i\in\Lambda} f(\hat \omega_i) d\hat\omega_i \tr P_{I_\delta}(\lambda V_\omega^\Lambda)\,,
\ee
where 
\be
I_\delta \ = \ \left[-\frac{|I|}{2}-(6+D_I),\frac{|I|}{2}-D_I\left(1+\delta\right)^{-1}\right]\,.
\ee
Since $\tr P_{I_\delta}(\lambda V_\omega^\Lambda) \le |\Lambda|$, we can use \eqref{eq:Lapprox} and \eqref{eq:rhoK} to  bound 
\begin{eqnarray*}
\E\ P_I(H_\omega^{\Lambda,\lambda})& \le & \frac{1}{\delta}\,\left(1+\delta\right)^{2+|\Lambda|}\,(D_I)^{-1}\, |I|\, |\Lambda|\   \int
\prod_{i\in\Lambda} f(\hat \omega_i) d\hat\omega_i \\
&&\hspace{-2cm} \le \   \frac{1}{\delta}\,\left(1+\delta\right)^{2+2|\Lambda|}\,\left(1+\left(2\,C\,K^{\frac{-\alpha}{1-\alpha}}\,+\,K\delta\,|J|\right)|J|\right)^{|\Lambda|}\,(D_I)^{-1}\, |I|\, |\Lambda| \\
& \le &  \frac{\, |I|\, |\Lambda|\exp\left\{(1+2K|J|^2)\,\left({2+|\Lambda|} \right)\,\delta\,+\,2\,C\,K^{\frac{-\alpha}{1-\alpha}}\,|J|\,|\Lambda|\right\}}{\delta\, D_I}\,,
\end{eqnarray*}
 where in the second step we have used \eqref{eq:rhobound}. Choosing 
\[K \ = \ \left(|J|\,|\Lambda|\right)^{\frac{1-\alpha}{\alpha}}\,;\quad \delta=\frac{1}{\left(1+|\Lambda|^{\frac{1-\alpha}{\alpha}}|J|^{\frac{1+\alpha}{\alpha}}\right)\,|\Lambda|}\,,\]
we get the desired bound \eqref{eq:wegn}.

\end{proof}

\begin{proof}[Proof of Lemma \ref{lem:aux_Wegner}]
\hspace{7cm} \vskip 3mm \noindent
The result follows from two consecutive applications of  Weyl's theorem, which states that for any pair of $n\times n$ hermitian matrices $C$ and $D$ we have 
\be\label{eq:Weyl}
\lambda_k(C)+\lambda_1(D) \ \le \ \lambda_k(C+D) \ \le \ \lambda_k(C)+\lambda_n(D)\,.
\ee
First we apply Weyl's theorem with $C=A$ and $D=x B$ to conclude that 
\be\label{eq:1we}\{k:\ \lambda_k(A+x B)\in I \mbox{ for some x }\in J\} \ \subset \ \{k:\ \lambda_k(A)\in \hat I\}\,.\ee  
Suppose now that $\lambda_k(A+x_0 B)\in I$ for some value $x_0$. Then using Weyl's theorem with $C=A+x_0 B$ and $D=(x-x_0) B$ we obtain that 
\be\label{eq:2we}
\lambda_k(A+x B)\notin I \mbox{ for } |x-x_0|>\alpha^{-1}\,|I|\,.
\ee
Combining  \eqref{eq:1we} and \eqref{eq:2we} 
we obtain the desired bound \eqref{eq:aux_Weg}.

\end{proof}

\section{Renormalization of tadpole contributions}\label{sec:renorm}
The aim of this section is to set up the appropriate resolvent expansion that will be used in the proofs of Lemmas
\ref{lemma:rep_matr} and \ref{lemma:rep}, namely to obtain  decompositions \eqref{eq:mdec_matr}
and \eqref{eq:mdec}. In particular, based on some combinatorial observation,
equation (\ref{eq:tadp_canc}), renormalization of tadpoles is done in both real 
(case $\mathcal N$) and momentum (case $\mathcal O$) space. 
The estimates that control various terms in the  resulting decompositions are established  in Section \ref{sec:repr}.
\vskip 3mm
\noindent We  decompose $H_\omega^\lambda$ as
\[H_\omega^\lambda=H_r+\tilde V\,,\quad
H_r:=-\frac{1}{2}\Delta-\sigma(p,E+i\epsilon)\,,\quad \tilde V:=\lambda V_\omega
+ \sigma(p,E+i\epsilon)\,,\]
where $\sigma(p,E+i\epsilon)$ is a solution of  \eqref{eq:self} for ($\mathcal O$) case.

Respectively, for  ($\mathcal N$) case we decompose
\[H_\omega^\lambda=H_r+\tilde V\,,\quad
H_r:=-\frac{1}{2}\Delta-\Sigma\,,\quad \tilde V:=\lambda V_\omega
+ \Sigma\,,\]
where $\Sigma$   is a periodic extension of sigma defined in \eqref{eq:self'}. Let
\be\label{eq:Rsigma}
R_r:=(H_r-E-i\epsilon)^{-1}\,.\ee
We can expand $R$ (defined in \eqref{eq:R}) into (truncated) resolvent series
\begin{equation}\label{eq:expa_matr}
R=\sum_{i=0}^N(-R_r\tilde V)^iR_r\ + \
(-R_r\tilde V)^{N+1}R\,.\end{equation}
To handle the renormalization of tadpole contributions properly, we
decide at which value of $n$ to halt the expansion in
\eqref{eq:expa_matr} individually for each contribution according to the
following rule (to which we will refer as a stopping rule):
If we open the brackets in \eqref{eq:expa_matr}, we obtain terms of the form
\[
R_r\theta R_r\theta\ldots R_r
\theta R_r
\]
where $\theta$ is either $-\lambda V_\omega$, or $-\sigma(p,E+i\epsilon)$ / $-\Sigma$
(whenever $\theta$ takes the later value we will refer to it as a {\em
bullet}).
Since $\sigma(p,E+i\epsilon)=O(\lambda^2)$, $\Sigma=O(\lambda^2)$ for all permissible values
of $E$, see Appendix \ref{sec:appendI}, one can unambiguously define
the {\em order} $l$ (in powers of $\lambda$) of the particular
contribution
\[R_r\theta R_r\theta\ldots R_r
\theta R_\sharp\,,\] (with $R_\sharp$ being either
$R_r$ or $R$) according to the following rule: Each factor of
$\sigma$ counts as $2$, while appearance of the random potential
counts as $1$, and we add up all the exponents to get the order of the
term. For instance, the order of the expression
\[R_r\sigma R_r\lambda V_\omega R_r
\sigma R\] is $5$.
To illustrate this procedure we
write down the expansion obtained in a case of $N=2$:
\begin{multline*}R\ = \ R_r \ - \
R_r\sigma R \
- \ \left\{\lambda R_r  V_\omega R \right\} \ = \\
 R_r \ - \ R_r\sigma R\ - \ \lambda
R_rV_\omega R_r \\ + \ \lambda R_r V_\omega R_r \sigma R \
  + \ \lambda^2 R_r
V_\omega R_r V_\omega R  \,,
\end{multline*}
 where the term in the curled brackets is the one
we expanded according to the stopping rule. Note that the
penultimate term is of order $3$. It is not difficult to see (see Lemma 3.1 in \cite{elgart} for the proof) that for a general $N$ we get
\begin{lemma}\label{lem:de}
For any integer $N$ we have a decomposition  
\be\label{eq:gra}
R\ = \ \sum_{l=0}^{N-1}
A'_lR_r+A'_NR+B_NR
\  = \ \sum_{l=0}^{N-1}A_l \ + \ \tilde
A_NR\,,
\ee
where $A'_0=I$, $A'_l$ is a summation over
all possible terms of the type
\begin{equation}\label{eq:ord}
R_r
\theta R_r\theta\ldots R_r \theta
\end{equation}
which are of the order $l>0$, while
\be\label{eq:ord1}
B_N\ = \ -\,
A_{N-1}\sigma\,.
\ee
The quantities $A_l$ and $\tilde A_N$ are defined as
\[A_l\ = \ A'_l R_r\,,\quad \tilde
A_N\ = \ A'_N+B_N\,.\]
\end{lemma}
In order to explain the renormalization, we borrow the following
paragraph from Section 3.1. of \cite{elgart} for notations.

For an integer $N$, let $\Upsilon_{N}$ be a set $\{1,\ldots,N,N+2,\ldots,2N+1\}$. Let $\Pi=\Pi_{N}$ be a set of partitions of $\Upsilon_{N}$
into disjoint subsets $S_j$ of cardinality $|S_j|\in 2\N$. Two
partitions $\pi = \{S_j\}_{j=1}^m$, $\pi' = \{S'_j\}_{j=1}^m$  are
equivalent, $\pi=\pi'$, if they coincide up to permutation. For
$S\subset \Upsilon_{N}$, let
\begin{equation}
\delta(x_{S})=\sum_{y\in\Z^3}\prod_{j\in S}\delta_{|x_j-y|}\,,
\end{equation}
where $\delta_{x}$, $x\in \Z$ is Kronecker delta function, and $x_S$
denotes the collection of $\{x_i\,,\ i\in S\}$.
 One has an identity (see e.g. \cite{Chen} Section 3.1 for details)
\begin{equation}\label{eq:tree_graph}
\E \left[\prod_{j\in \Upsilon_{N,N}}
\omega_{x_j}\right]=\sum_{m=1}^N\sum_{\pi=\{S_j\}_{j=1}^m}\,
\prod_{j=1}^m c_{|S_j|}\delta(x_{S_j})\,,
\end{equation}
where $c_{2l}\le (cl)^{2l+1}$ and
$c_2=\E\, \omega_x^2=1$, provided assumption ($\mathcal A$) holds. The set $S_j$ in the partitions
$\pi\in \Pi$ can be of the special type: If
\begin{equation}\label{eq:gatedef}
 S_j=\{i,i+1\}
\end{equation}
we will refer to it as a {\em tadpole}, or a {\em gate} set. 

Now let $\pi_k^c$ denote a collection of disjoint sets $\{S_j\}$ such that
any $S_j\in\pi_k^c$ is a tadpole, and the cardinality of $\pi_k^c$
is $k$. Then any partition $\pi$ can be decomposed as
$\pi=\pi_k^c\cup\{S\}$ for some $0\le k \le N$, where $S$ satisfies
$(\cup_{s_j\in\pi_k^c}S_j)\cup S=\Pi_N$. Note that we didn't require
$S$ to be a tadpole free set. We will denote by $\pi_0$ a partition
of $\Upsilon_N$ such that no $S_j\in\pi_0$ is a tadpole. Lemmas
\ref{lem:renorm_matr} and \ref{lem:renorm}  below hinge on the
following observation:
 \begin{multline}\label{eq:tadp_canc}
 \sum_{k=0}^N(-1)^k\, \sum_{\substack{\pi\in \Pi:\\ \pi=\pi_k^c\cup \{S\}}} \ \E  \left[\prod_{i\in S}
\omega_{x_i}\right]\ \prod_{S_l\in\pi_k^c}\delta(x_{S_l}) \ =\ \sum_{\substack{\pi\in \Pi:\\ \pi=\pi_0}}\
\prod_{S_j\in \pi} c_{|S_j|}\delta(x_{S_j})\,.
 \end{multline}
Note that the summation on the right hand side runs over the tadpole-free partitions. To verify \eqref{eq:tadp_canc} one just need to make a straightforward check that all tadpole contributions on the left hand side cancel out exactly.

To see why the renormalization of tadpole's contribution in conjuncture with the above stopping  procedure is useful in both ($\mathcal O$) and ($\mathcal N$) cases, we consider two different tracks for each one of them:
\subsection{Case ($\mathcal N$)}

Let $P_x$ denote the projection onto the set $\hat \Theta-x$ for $x\in k\Z^3$, where $\hat \Theta,\ k$ are introduced in Assumption ($\mathcal N$). Then we can partition the identity operator as
\[I \ = \ \sum_{x\in k\Z^3} P_x\,,\]
and $P_xP_y=0$ for $x\neq y$.  Instead of estimating the matrix elements of the corresponding terms in expansion \eqref{eq:gra} directly, we will consider norms of operators $P_x R P_y$ for $x,y\in k\Z^3$. Clearly, if $x'\in Range\, P_x$ and   $y'\in Range\, P_y$, then
\[|R(x',y')| \ \le \ \|P_x R P_y\|\,.\]
To evaluate $P_x R P_y$ we use Lemma \ref{lem:de}. We insert the partitions of identity between each factor of the resolvent in \eqref{eq:ord}, with the net result
\begin{multline}\label{eq:Alord} P_{x_0}A_l P_{x_{n+1}} \ = \\
\sum_{\substack{\theta,\,x_j\in\Z^3;\\ j=1,\ldots,n}}P_{x_0} R_r P_{x_1}\theta(x_1)P_{x_1} R_r P_{x_2}\theta(x_2)...P_{x_{n-1}} R_r P_{x_n}
\theta(x_n)P_{x_{n}} R_r P_{x_{n+1}}
\end{multline}
where $\theta(x)$ is either $-\lambda\omega_{x}
\mathcal D$, or $-\Sigma$ (defined in Eqs. \ref{eq:mathcalD} and \ref{eq:Sigm}, respectively). The index $n$ here depends on the particular contribution in $A'_l $, but the order of all contributions is $l$.

To estimate the typical size of $P_x A_l P_y$ we consider the matrix
\be\label{eq:matrAl}
{\mathcal A}_{x,y} \ := \ \E\left\{P_x A_l P_y\cdot P_y A^*_l P_x\right\}\,.
\ee
The key technical lemmas are the following assertions:
\begin{lemma}\label{lem:renorm_matr}
We have
\begin{multline}\label{eq:gaf1_matr}
{\mathcal A}_{x,y}\ = \ \lambda^{2l}\, \sum_{\substack{\pi\in \Pi_l:\\ \pi=\pi_0}}\
\sum_{{\substack{x_j\in k\Z^3: \\ j\in\Upsilon_{l}}}}\ \prod_{S_j\in \pi}
c_{|S_j|}\delta(x_{S_j}) \\ \times \  \prod_{i=0}^{l-1}\left\{P_{x_i}R_rP_{x_{i+1}}\,{\mathcal D}\right\}\,R_rP_yR^*_r\,\prod_{i=l+2}^{2l+1}\left\{{\mathcal D}\,P_{x_i}R^*_rP_{x_{i+1}}\right\}
 \,,
\end{multline}
where  we  are using  convention
$x_0=x$; $x_{2l+2}=y$.
\end{lemma}

\begin{proof}
\hspace{7cm} \vskip 3mm \noindent
 We first observe that by definition of $ {\mathcal D}$, $\Sigma$ and \eqref{eq:self'} we have
 \[\lambda^2 {\mathcal D}P_x R_r P_x {\mathcal D} \ = \ P_x\Sigma\,.\]
Using this identity in the definition of $A_l$, we can represent
\begin{multline}\label{eq:reprAno}
P_x A_l P_y\cdot P_y A^*_l P_x \ = \    \lambda^{2l}\, \sum_{k=0}^N(-1)^k\
\sum_{{\substack{x_j\in k\Z^3: \\ j\in\Upsilon_{l}}}}\
\sum_{\substack{\pi\in \Pi:\\ \pi=\pi_k^c\cup \{S\}}} \   \left[\prod_{i\in S}
\omega_{x_i}\right]\ \prod_{S_l\in\pi_k^c}\delta(x_{S_l})
\\ \times \  \prod_{i=0}^{l-1}\left\{P_{x_i}R_rP_{x_{i+1}}\,{\mathcal D}\right\}\,R_rP_yR^*_r\,\prod_{i=l+2}^{2l+1}\left\{{\mathcal D}\,P_{x_i}R^*_rP_{x_{i+1}}\right\}\,.\end{multline}
Computing the expected value of the left and right hand sides with respect to randomness and using \eqref{eq:tadp_canc}, we obtain  \eqref{eq:gaf1_matr}.

\end{proof}

\subsection{Case ($\mathcal O$)}
The counterpart of Lemma \ref{lem:renorm_matr} in this case is the following  generalization of Lemma 3.2 of \cite{elgart} (applicable for the non correlated randomness, i.e. $\hat u(p)=1$). In what follows, we will use the short hand notation $E(p)$ in place of $e(p)-E-i\epsilon-\sigma(p,E+i\epsilon)$, and $E^*(p)$ for the hermitian conjugate of the multiplication operator $E(p)$. The renormalized propagator $R_r$ in this case will be given by its kernel 
\begin{equation}\label{eq:R_r_O}R_r(z,w)\ = \ \int_{\T^3}e^{i2\pi(z-w)p}\frac{d^3p}{E(p)}\,.\end{equation}

The following assertion holds:
\begin{lemma}\label{lem:renorm}
For $A_l$ defined in Lemma \ref{lem:de}, the function $\E\, |A_l(x,y)|^2$ is a function of the variable $x-y$. Let
\be\label{eq:A_Lo} {\mathcal A}_{l,E}(x-y) \ :=\ \E\, |A_l(x,y)|^2\,,\ee
then we have
\begin{multline}\label{eq:gaf}
{\mathcal A}_{l,E}(x-y)\ = \ \lambda^{2l}\,\int_{(\T^3)^{2l+2}}  \, 
e^{i\alpha}\ \frac{dp_{l+1}}{E(p_{l+1})}\,\frac{dp_{2l+2}}{E(p_{2l+2})}\
\prod_{j=1}^{l} \, \frac{dp_j}{E(p_j)}\  \prod_{j=l+2}^{2l+1} \,\frac{dp_j}{E^*(p_j)}\\ \times \  \prod_{i\in \Upsilon_l}
\hat u(p_j-p_{j+1})\
\sum_{\substack{\pi\in \Pi_l:\\ \pi=\pi_0}} \  \prod_{S_k\in\pi}
c_{|S_k|}\,\delta\left(\sum_{i\in S_k}p_i-p_{i+1}\right) \,,
\end{multline}
where
\[ \alpha\ :=\ -i2\pi(p_1+p_{l+2})\cdot(x-y)\,.\]
\end{lemma}
\begin{proof}
\hspace{7cm} \vskip 3mm \noindent
Let $V^\delta_\omega$ be a random potential of the form
\[
V^\delta_\omega(x) \ = \ \sum_{i\in k\Z^3} \omega_i e^{-\delta|i|}\, u(x-i)\,.
\]
Then $V^\delta_\omega \longrightarrow V_\omega$ in the strong operator topology  as $\delta$ converges to $0$. Similarly, we can define quantities $H^{\lambda,\delta}_\omega$, $R^\delta$, $R_r^\delta$, and $A_l^\delta$ by replacing $V_\omega$ with $V^\delta_\omega$. One can readily check that
\[R^\delta(x,y)\rightarrow R(x,y)\,;\quad R^\delta_r(x,y)\rightarrow R_r(x,y)\,,\quad A_l^\delta(x,y)\rightarrow A_l(x,y)\]
in the limit $\delta\rightarrow0$.
The advantage of working with the regularized random potential is due to the fact that it is summable and therefore admits Fourier transform. Namely, we have
\[\hat V^\delta_\omega( p)\ =\  \hat u(p)\,\hat \omega_\delta(p)\,,\]
where
\[\hat \omega_\delta(p) \ := \ \sum_{n\in\Z^3}e^{-i2\pi p\cdot
 n}\,\omega_n\,e^{-\delta|n|}\,.\]
Since by  \eqref{eq:self} we have
 \[\lambda^2 \int_{\T^3}|\hat u(p-q)|^2\,R_r(q)dq  \ = \ \sigma(p,E+i\epsilon)\,,\]
we can express $|A^\delta_l(x,y)|^2$ (analogously to \eqref{eq:reprAno}) as
\begin{multline}\label{eq:tree_graphdelta}
 |A^\delta_l(x,y)|^2 \ = \\    \lambda^{2l}\ \int_{(\T^3)^{2l+2}}  \, e^{i\beta}\, \frac{dp_{l+1}}{E(p_{l+1})}\,\frac{dp_{2l+2}}{E^*(p_{2l+2})}\
\prod_{j=1}^{l} \, \frac{dp_j}{E(p_j)}\  \prod_{j=l+2}^{2l+1} \,\frac{dp_j}{E^*(p_j)}\ \prod_{i\in \Upsilon_l}
\hat u(p_j-p_{j+1})\\  \times \  \sum_{k=0}^N(-1)^k\ \sum_{\substack{\pi\in \Pi_l:\\ \pi=\pi_k^c\cup \{S\}}} \   \prod_{i\in S}
\hat \omega_\delta(p_i-p_{i+1})\ \prod_{S_l\in\pi_k^c}\delta\left(\sum_{i\in S_l}p_i-p_{i+1}\right)
\,,\end{multline}
where $\beta:= 2\pi\{-(p_1+p_{l+2})\cdot x
+(p_{l+1}+p_{2l+2})\cdot y\}$.  It follows from \eqref{eq:tadp_canc} that
 \begin{multline*}
 \sum_{k=0}^N(-1)^k\, \sum_{\substack{\pi\in \Pi_l:\\ \pi=\pi_k\cup \pi_k^c}} \ \prod_{S_l\in\pi_k^c}\delta\left(\sum_{i\in S_l}p_i-p_{i+1}\right)\  \E  \left[\prod_{i\in S_j\in \pi_k}
\hat \omega_\delta(p_i-p_{i+1})\right]\\ \overset{d}{\longrightarrow}\ \sum_{\substack{\pi\in \Pi:\\ \pi=\pi_0}}\
\prod_{S_j\in \pi} c_{|S_j|}\delta\left(\sum_{i\in S_j}p_i-p_{i+1}\right)\,,
 \end{multline*}
where $\overset{d}{\longrightarrow}$ stands for the convergence (with respect to  $\delta$) in the distributional sense. Therefore, taking the expected value on the both sides of \eqref{eq:tree_graphdelta} as well as  $\delta\rightarrow0$ limit (where we use the smoothness of the integrand), we arrive to the expression  that coincides with \eqref{eq:gaf}, up to the prefactor $e^{i\beta}$ instead of $e^{i\alpha}$ in the integrand. But the product of the delta functions allows to replace $\beta$ with $\alpha$ (see Subsection \ref{sub:F} below), hence the result.

\end{proof}

\section{Proof of Lemmas \ref{lemma:rep_matr} and \ref{lemma:rep}}\label{sec:repr}
\begin{proof}[Proof of Lemma \ref{lemma:rep_matr}]
\hspace{7cm} \vskip 3mm \noindent
We observe that for any $x'\in Range\, P_x$ and   $y'\in Range\, P_y$
\be\label{eq:albnd}\E\, |A_l(x',y')|^2\ \le \ \|
{\mathcal A}_{x,y}\|_1\,,\ee
where ${\mathcal A}_{x,y}$ is defined in \eqref{eq:matrAl} and $\|\cdot\|_1$ stands for  the maximum absolute column sum norm. Indeed,
\[
 |A_l(x',y')|^2 \ = \ \langle x'| P_x A_l \ |y'\rangle\langle y'|\  A_l^* Px|x'\rangle \ \le \
\langle x'| P_x A_l P_y A_l^* Px|x'\rangle\,,\]
hence
\[\E\, |A_l(x',y')|^2\ \le \ \E\, \langle x'| P_x A_l P_y A_l^* Px|x'\rangle \ \le \ \|
{\mathcal A}_{x,y}\|\ \le \ \|
{\mathcal A}_{x,y}\|_1\,.\]
Therefore, using Lemma \ref{lem:renorm_matr}, we obtain that for such $x'$ and $y'$
\begin{multline}\label{eq:Al_matrbnd}
\E\, |A_l(x',y')|^2\ \le \  \lambda^{2l}\, \|{\mathcal D}\|^{2l}\, \sum_{\substack{\pi\in \Pi_l:\\ \pi=\pi_0}}\
\sum_{{\substack{x_j\in k\Z^3: \\ j\in\Upsilon_{l}}}}\ \prod_{S_j\in \pi}
c_{|S_j|}\delta(x_{S_j}) \\ \times \  \prod_{i=0}^{l}\left\|P_{x_i}R_rP_{x_{i+1}}\right\|_1\  \prod_{i=l+1}^{2l+1}\left\|\,P_{x_i}R^*_rP_{x_{i+1}}\right\|_1
 \,,
\end{multline}
with convention $x_{l+1}=y$, $x_0=x_{2l+2}=x$.

Now we are in position to use Lemma \ref{lem:propSigma} to bound the products of the norms on the right hand side of \eqref{eq:Al_matrbnd}, for all  $|\epsilon|<\kappa/2$ defined in this lemma and all $E<E_0$ with $E_0$ be given by \eqref{eq:Esigm} as
\be\label{eq:bndGA}
 \prod_{i=0}^{l}\left\|P_{x_i}R_rP_{x_{i+1}}\right\|_1\  \prod_{i=l+1}^{2l+1}\left\|\,P_{x_i}R^*_rP_{x_{i+1}}\right\|_1\ \le \ |\hat \Theta|^{2l+2}\ \prod_{i=0}^{2l+1} G_{E^*}(x_i,x_{i+1})\,,
\ee
where $G_{E^*}(x,y)$ is defined in \eqref{eq:G(xy)} and
$E^*=E^*_{\mathcal N}$ is given in the statement of Lemma
\ref{lemma:rep_matr}. We remind the reader that $G_{E^*}(x,y)$ is
positive for any $x,y\in\Z^3$.

Plugging \eqref{eq:bndGA} into \eqref{eq:Al_matrbnd} we get the estimate
\begin{multline}\label{eq:Al_matrbnd'}
\E\, |A_l(x',y')|^2\\ \le \  \lambda^{2l}\, \ |\hat \Theta|^{2l+2}\,\|{\mathcal D}\|^{2l}\, \sum_{\substack{\pi\in \Pi_l:\\ \pi=\pi_0}}\
\sum_{{\substack{x_j\in k\Z^3: \\ j\in\Upsilon_{l}}}}\ \prod_{S_j\in \pi}
c_{|S_j|}\delta(x_{S_j}) \ \prod_{i=0}^{2l+1} G_{E^*}(x_i,x_{i+1})\\ \le \ \lambda^{2l}\, \ |\hat \Theta|^{2l+2}\,\|{\mathcal D}\|^{2l}\, \sum_{\substack{\pi\in \Pi_l:\\ \pi=\pi_0}}\
\sum_{{\substack{x_j\in\Z^3: \\ j\in\Upsilon_{l}}}}\ \prod_{S_j\in \pi}
c_{|S_j|}\delta(x_{S_j}) \ \prod_{i=0}^{2l+1} G_{E^*}(x_i,x_{i+1})
 \,.
\end{multline}
The latter expression, however, coincides (up to the factor $\ |\hat \Theta|^{2l+2}\,\|{\mathcal D}\|^{2l}$) with the corresponding term for the random potential of the form
\[\tilde V_\omega(x) \ := \ \sum_{i\in \Z^3} \omega_i\,,\]
investigated in \cite{elgart} (c.f. Eq. 3.15 there). As a result, the bound \eqref{eq:l_matr} follows from Lemma 1.1 of \cite{elgart}.

To get \eqref{eq:lt_matr} note that it follows from Lemma \ref{lem:de}
that
\[\tilde A_N \ = \ A_N(H_r-E-i\epsilon) \ - \
A_{N-1} \Sigma\,.\]
We therefore obtain
\begin{multline}
\E \,|\tilde A_N(x,y)| \ \le \ \sum_{z\in\Z^3}\left\{\,\left(\E
\,|A_N(x,z)|^2\right)^{1/2}\,\cdot\,
|(H_r-E-i\epsilon)(z,y)|\right.
\\ \left.
+ \ \left(\E
\,|A_{N-1}(x,z)|^2\right)^{1/2}\,\cdot\,|\Sigma(z,y)|\,\right\}
\\ < \
\sum_{\substack{z\in\Z^3:\\ |z-y|\le 1}}\left(\E
\,|A_N(x,z)|^2\right)^{1/2} \
+ \ 2\,|\hat \Theta|\,\lambda^2\,\|D\|^2\,\sum_{\substack{z\in\Z^3:\\ |z-y|\in\hat \Theta}}\, \left(\E
\,|A_{N-1}(x,z)|^2\right)^{1/2}
\,,\nonumber
\end{multline}
provided $\lambda$ is sufficiently small, and where in the last step we have used \eqref{eq:bndSi}.

It now readily follows from \eqref{eq:l_matr}  that for $\lambda$ is sufficiently small, the right hand side of the above equation is
bounded by
\[C'\,\ |\hat \Theta|^{l+1}\,\|{\mathcal D}\|^{l}\,\sqrt{(4N)!\,E^*}
\,\left(C\ \ln^{9}E^*\ \frac{\lambda^{2}}{\sqrt{E^*}}\right)^{N/2}
\,e^{-\sqrt\frac{{E^*}}{12}\,|x-y|}\,,\]
with some generic constant $C'$. As a result, we have proved
\eqref{eq:lt_matr}.

Further, \eqref{eq:stand_matr} follows from the fact that $A_0(x,y)=R_r(x,y)$ and
\[|R_r(x,y)|\ \le\ G_{E^*}(x,y)\]
 by Lemma \ref{lem:propSigma}.  But the application of Lemma \ref{lem:freeGF} then shows the validity of  \eqref{eq:stand_matr}.

\end{proof}

\subsection{Feynman graphs}\label{sub:F}
At this point we have to introduce some additional notation:

\begin{defn} \label{def:equiv} We consider products of delta functions
with arguments that are linear combinations of the momenta $\{ p_1,
p_2,\ldots , p_{2 n +2}\}$. Two products of such delta functions
are called  {\em equivalent} if they determine the same affine
subspace of $\T^{2 n+2} = \{ p_1, p_2,\ldots , p_{2 n +2}\}$.
\end{defn}

One can obtain new delta functions from the given ones, by taking
linear combinations of their arguments. In particular, we can obtain
identifications of momenta.

\begin{defn}  \label{def:forced} The product $\delta (\sum_j a_j
p_j)$ of delta functions
$\Delta_\pi$  {\em forces} a new delta function, if $\sum_j a_j p_j =0$ is an identity in the affine subspace
determined by $\Delta_\pi$.
\end{defn}

One can readily see that in the integrand of rhs of (\ref{eq:gaf})
one has a forced delta function
$\delta(p_1-p_{l+1}+p_{l+2}-p_{2l+2})$, the fact used in Lemma \ref{lem:renorm}.

$A_{l,E^*}(x-y)$ is conveniently  interpreted in terms of the so
called Feynman graphs (the pseudograph, to be precise, since loops
and multiple edges are allowed here). The graph, associated with
particular partition $\pi$ of $\Upsilon_{n,n}$ is constructed
according to the following rules (see Figure \ref{fig:con} and
\ref{fig:conII}): We first draw two line segments, each containing
$n$ vertices (elements of $\Upsilon_{n,n}$).
 The vertices are joined by directed edges (momentum lines) representing the corresponding momenta:
$p_1, \ldots , p_{n+1}$ and $p_{n+2}, \ldots , p_{2n+2}$. To each
line $p_j$ we assign a propagator $F(p_j)$, with some given function
$F$, save momentum lines $p_1$ and $p_{n+2}$, which carry additional
phases $e^{-i2\pi p_1\cdot (x-y)}$ and $e^{-i2\pi p_{n+2}\cdot
(x-y)}$, respectively. For $\pi=\{S_j\}_{j=1}^m$ we identify all
vertices in each subset $S_j$ as the same vertex (in  Figure
\ref{fig:con}, the paired vertices are connected by  dashed lines).
\begin{figure}[ht]
\includegraphics[trim = 0mm 0mm 0mm 10mm, clip, width=10cm]{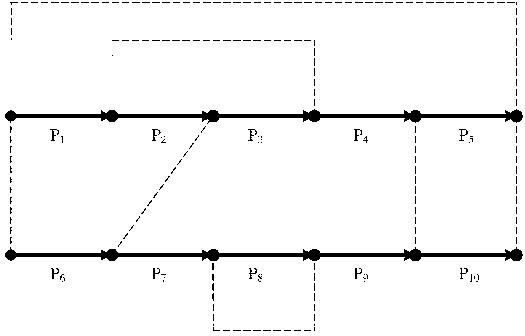}
\caption{Construction of the Feynman graph, part I, $n=4$. The
corresponding delta functions are $\delta(p_1-p_2+p_3-p_4)$,
$\delta(p_4-p_5+p_9-p_{10})$, $\delta(p_2-p_3+p_6-p_7)$, and
$\delta(p_7-p_9)$. The last delta corresponds to the tadpole. Note
that the sum of all momenta in the above delta functions gives a
forced delta function $\delta(p_1-p_5+p_6-p_{10})$, hence we can
introduce the dashed lines connecting vertices $1,6,7,$ and $12$,
identifying them as a single vertex.} \label{fig:con}
\end{figure}
Note that thanks to the existence of the forced delta function
$\delta(p_1-p_{l+1}+p_{l+2}-p_{2l+2})$, we can identify vertices
$1,l,l+1,2l$ as a single one, and therefore one can think about the
closed graph (with special rules that apply for momentum lines $p_1$
and $p_{l+2}$, mentioned above). To summarize, the outcome of this
construction is a directed closed graph, which is called the Feynman
graph associated with the partition $\pi$. The momenta in the graph
satisfy the Kirchhoff's first law, that is the total momenta
entering into each internal vertex add up to zero (if arrow faces
outward the vertex, we count its momentum with a minus sign). A
tadpole corresponds to the so-called {\em $0$-loop}, that is some
(directed) line of the graph claims one vertex as its both
endpoints. For a given Feynman graph $G$, one can choose a
particularly useful expression for the product of delta functions
$\Delta_\pi$. Choose any spanning tree of $G$ which does not contain
momentum lines $p_1,p_{l+2}$. The edges belonging to the spanning
tree will be called the {\em tree} edges (momentum lines), and all
the rest are the {\em loop} edges (since an addendum of any loop's
momentum line creates a loop). Let us enumerate the tree variables
as $u_1,...,u_k$, and loop variables as $w_1,...,w_n$, with say
$w_1=p_1,w_2=p_{l+2}$  (note that $k+n=2l+2$). The number $k$ of the
tree momenta coincides with the number of the delta functions in
$\Delta_\pi$.

\begin{figure}[ht]
\includegraphics[trim = 0mm 0mm 0mm 55mm, clip, width=10cm]
{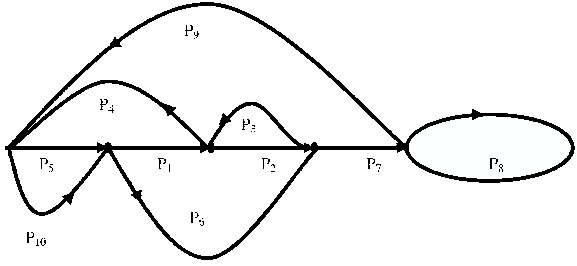}
\caption{Construction of the Feynman graph, part II: Identification
of the vertices. The tadpole corresponds here to $0$-loop.}
\label{fig:conII}
\end{figure}
 One can check (see e.g. \cite{EY}) that the product of
delta functions $\Delta_\pi$ is equivalent to
\begin{equation}\label{eq:deltapr}
\prod_{i=1}^k \delta(u_i-\sum_{j=1}^la_{ij}w_j)\,,
\end{equation}
with
\[
a_{ij}:= \begin{cases} \pm1 & \hbox{ loop that contains $u_i$ is created by adding $w_j$
to the spanning tree}  \cr
 0
 &  \hbox{
otherwise} \cr
\end{cases}\,.
\]
The choice of the sign depends on the mutual orientation of $u_i$
and $w_j$.

\begin{proof}[Proof of Lemma \ref{lemma:rep}]
\hspace{7cm} \vskip 3mm \noindent
We establish the exponential decay of ${\mathcal A}_{l,E}(x-y)$ in $|x-y|$  from the following analytic argument. To this end we generalize  the arguments used in Section 4. of \cite{elgart} to the {\it momentum--dependent} 
self energy $\sigma$ (defined in Eq.~\ref{eq:self}).

We start with some additional notation: Let $\mathcal R$ denote a  rectangle formed by the points
\[\{\ -1/2\pm i\sqrt \delta;\
1/2\pm i\sqrt \delta\}\,,\]
where the parameter $\delta=\delta_{\mathcal O}$ was introduced in
\eqref{eq:deltaE}. For a unit vector $e\in\Z^3$, we will decompose
$\T^3\ni p=p\cdot e\oplus p^\perp$, where  $p^\perp\in\T^2$. In what
follows we will use the norm $\|\cdot\| _{\infty,\mathcal R}$
defined as
\be\label{eq:normr} \|f\| _{\infty,{\mathcal R}}\ :=\  \max_{e}\sup_{p\cdot e\in{\mathcal R},\ p^\perp\in\T^2}|f(p)|\,.\ee
We note that for $u$ satisfying \eqref{eq:condonu} and $\delta$ sufficiently small, one has
\be\label{eq:condu'} \|\hat u\|_{\infty,{\mathcal R}} \ \le \ \|\hat u\|_{\infty}\,+\,CA^{-4} \sqrt{\delta} \,.\ee

We will show that for a general value of $l$,
\begin{equation}\label{eq:expdec}
{\mathcal A}_{l,E}(x)\ \le \ \|\hat u\|_{\infty,{\mathcal R}}^{2l}\cdot e^{-\sqrt{\delta/3}\,|x|}\ \hat {\mathcal A}_{l,E^*}(0)\,,
\end{equation}
where
\begin{multline}\label{eq:hatmathcalA}
\hat {\mathcal A}_{l,E^*}(0)\ :=\  \lambda^{2l}\,\int_{(\T^3)^{2l+2}}  \, e^{i\alpha}\ \frac{dp_{l+1}}{e(p_{l+1})+E^*}\,\frac{dp_{2l+2}}{e(p_{2l+2})+E^*}\
\prod_{j\in\Upsilon_l} \, \frac{dp_j}{e(p_j)+E^*}\\  \times \
\sum_{\substack{\pi\in \Pi_l:\\ \pi=\pi_0}} \  \prod_{S_k\in\pi}
c_{|S_k|}\,\delta\left(\sum_{i\in S_k}p_i-p_{i+1}\right)\,,
\end{multline}
with $\alpha$ defined in Lemma \ref{lem:renorm}.

The expression $\hat {\mathcal A}_{l,E^*}(0)$ has been  studied in Section 4 of \cite{elgart}. It was shown there that
\be\label{eq:A0}
\hat {\mathcal A}_{l,E^*}(0) \ \le \ (4l)!\, E^*\,\left(C\ln^9(E^*)\
\frac{\lambda^{2}}{\sqrt{E^*}}\right)^l\,.
\ee
Combining \eqref{eq:expdec}, \eqref{eq:A0} and \eqref{eq:condu'},  we obtain \eqref{eq:l}. Also, since  $A_0$ was defined to be equal $R_r$ in Lemma \ref{lemma:rep}, we get $\hat {\mathcal A}_{0,E^*}(0)=\left(G_{E^*}(0,0)\right)^2 $. Equation (\ref{eq:stand}) follows from the bound \eqref{eq:freeGF} on the free Green function.

The relation \eqref{eq:lt} is obtained analogously to the derivation of \eqref{eq:lt_matr} in the proof of Lemma \ref{lemma:rep_matr}.

\par
To prove that \eqref{eq:expdec} holds true let us choose  for any given $x\in \Z^3$  the index $\gamma\in\{1,2,3\}$ such that
\begin{equation}\label{eq:ga}
|x_\gamma|=\max_{i\in\{1,2,3\}} |x_i|\,.
\end{equation}
Then $|x_\gamma|\ge|x|/\sqrt3$. In order to extract the
exponential decay of ${\mathcal A}_{l,E}(x)$ we first perform the integration
in the right hand side of \eqref{eq:gaf} over the tree momenta, using
(\ref{eq:deltapr}).

Let us use the shorthand notation $\sum_\pi$ for
a sum over all possible partitions in \eqref{eq:gaf}, $c_\pi$
for a product of the corresponding $c_{Sj}$, $r_\pi$ for
the number of the delta functions containing the loop momentum
$w_1$ in the $\pi$'s partition, and $s_\pi$ will denote the number of
$\hat u$ terms involving $w_1$ after the integration of tree momentums.

Let $E(p)=e(p)-E-i\epsilon-\sigma(p, E+i\epsilon)$. We have
\begin{multline}\label{eq:decompal}
{\mathcal A}_{l,E}(x)\ = \ \lambda^{2l} \sum_\pi c_\pi\int
dw_1\,e^{-i2\pi
w_1\cdot x}\ \prod_{i=1}^{r_\pi} \,\frac{1}{E^\sharp(w_1+q_i)}\prod_{j=1}^{s_\pi}\hat{u}(w_1+Q_j)
\\ \times \ \int \e^{-i2\pi w_{2}\cdot x}\ \prod_{t\in \Phi'}\,dw_t\ \prod_{i=r_\pi+1}^{2n+2}\ \,\frac{1}{E^\sharp(q_i)}
\prod_{j=s_\pi+1}^{2n}\hat{u}(Q_j) \,,
\end{multline}
where $E^{\sharp}(p)$ stands for either $E(p)$ or $E^*(p)$, $\Phi'$ is a set of indices of loop momentum that does not
include $w_1$, and  $q_i$, $Q_j$ are some linear combinations of the loop
variables in $\Phi'$. Note now that
\begin{multline}
\int dw_1\prod_{i=1}^{r_\pi} \,\frac{1}{E^\sharp(w_1+q_i)}\,
e^{-i2\pi w_1\cdot x}\prod_{j=1}^{s_\pi}\hat{u}(w_1+Q_j)\\
= \ \int dw_1^\perp \,e^{-i2\pi (w_1\cdot x-(w_1\cdot e_\gamma) x_\gamma)}\\
\times \int_{-1/2}^{1/2} d(w_1\cdot e_\gamma)\ \prod_{i=1}^{r_\pi}
\ \frac{1}{E^\sharp(w_1+q_i)}\,e^{-i2\pi (w_1\cdot e_\gamma) x_\gamma}\ \prod_{j=1}^{s_\pi}\hat{u}(w_1+Q_j)\,.
\end{multline}

Without loss of generality, let us assume
that $x_\gamma>0$. The integrand as
a function of $w_1\cdot e_\gamma$ is $1$-periodic, analytic inside
the rectangle ${\mathcal R}_-:=\C_-\cap{\mathcal R}$  for sufficiently small $E^*$. Moreover, we have
\be\label{eq:bndonE}{\rm Re}\, e(p-i\sqrt \delta e_\gamma )\ \ge \ e(p)-2\pi^2\delta\ee
uniformly in $p\in\T^3$, provided $\epsilon$ is sufficiently small, where we have used the definition \eqref{eq:e(p)} and
\be\label{eq:sinext}
\sin(a+ib)\ = \ \sin a \cosh b +i \cos a \sinh b\,.\ee
Combining this bound with  \eqref{eq:bndonsigmcomp}, we get
\be\label{eq:maxEp}
\min\left(|E(p)|,|E( p-i\sqrt \delta e_\gamma)|\right) \ > \ e(p)\,+\,E^*\,, p\in\T^3\,.\ee
Moreover, the
periodicity of the integrand implies that the integrals over the vertical
segments of ${\mathcal R}_-$ coincide:
\begin{multline}
\int_{-1/2}^{-1/2-i\sqrt \delta} d(w_1 \cdot
e_\gamma)\prod_{i=1}^{r_\pi} \,\frac{1}{E^\sharp(w_1+q_i)}\,
e^{-i2\pi x_\gamma\, (w_1\cdot e_\gamma )}
\prod_{j=1}^{s_\pi}\hat{u}(w_1+Q_j)
\\ = \
\int_{1/2}^{1/2-i\sqrt \delta} d(w_1\cdot e_\gamma)\prod_{i=1}^{r_\pi}
\,\frac{1}{E^\sharp(w_1+q_i)}\,e^{-i2\pi x_\gamma\, (w_1\cdot e_\gamma )}\,\prod_{j=1}^{s_\pi}\hat{u}(w_1+Q_j).
\end{multline}
Therefore
\begin{multline}
\left|\int_{\T} d(w_1\cdot e_\gamma)\prod_{i=1}^{r_\pi}
\,\frac{1}{E^\sharp(w_1+q_i)}\,e^{-i2\pi x_\gamma\, (w_1\cdot e_\gamma )}\prod_{j=1}^{s_\pi}\hat{u}(w_1+Q_j)\right| \\
= \ \left|\int_{\T-i\sqrt \delta}
d(w_1\cdot e_\gamma)\prod_{i=1}^{r_\pi} \,\frac{1}{E^\sharp(w_1+q_i)}\,
e^{-i2\pi x_\gamma\, (w_1\cdot e_\gamma )}\prod_{j=1}^{s_\pi}\hat{u}(w_1+Q_j)\right|\\
\le \ \|\hat u\|_{\infty,\mathcal R}^{s_\pi}\cdot e^{-x_\gamma\sqrt{\delta}}\int_{\T}
d(w_1\cdot e_\gamma)\prod_{i=1}^{r_\pi} \,
\left|\frac{1}{E\left(w_1+q_i-ie_\gamma\sqrt \delta\right)}\right|\\
\le \ \|\hat u\|_{\infty,\mathcal R}^{s_\pi}\cdot e^{-|x|\sqrt{E^*/3}}\,\int_{\T} d(w_1\cdot
e_\gamma)\prod_{i=1}^{r_\pi} \,\frac{1}{e(w_1+q_i)+E^*}\,,
\end{multline}
where in the last step we have used \eqref{eq:maxEp}.
Using the estimate \eqref{eq:bndonsigmcomp} again, we also have
\[|E(p)|\ > \ e(w_1+q_i)+E^*\,, \quad p\in\T^3\,.\]

Putting everything together on the right hand side of \eqref{eq:decompal}, we get the bound \eqref{eq:expdec}.

\end{proof}


\appendix
\section{Bounds on the free Green function}\label{sec:append} The
 free Green function $G_E(x,y)$ was defined in \eqref{eq:G(xy)}. We have
\begin{lemma}\label{lem:freeGF}
Define the function
$\psi_\alpha\in l_2(\Z_+)$ as
\[\psi_\alpha(r) \ = \ \,e^{-r\frac{\sqrt{-E}}{\alpha}} \,
\max\left((-E)^{(d-2)/2}\,,\,(1+r)^{(2-d)}\right)\,.\] Then for
$d\ge3$  and $-1<E<0$ we have
\begin{equation}\label{eq:freeGF}
0\ < \ G_E(x,y)\ <\ C_d\, \,\psi_{3d}(|x-y|)\,,
\end{equation}
for all $x,y\in\Z^d$.
\end{lemma}

\begin{remark}
A similar statement is known to hold on $\R^d$, \cite{Salmhofer1999}.
We are not aware of its lattice version in the existing literature.
The positivity of $G(x,y)$ on the lattice is well known, so it is an
upper bound we are after here.
\end{remark}

\begin{proof}
\hspace{7cm} \vskip 3mm \noindent
In what follows, we will use the following properties of the
function $\psi$ for $d\ge3$ and $E^*$ sufficiently small:
\begin{enumerate}
\item[(a)] \[\|\psi_\alpha\|_\infty \ = \ 1\,;\]
\item[(b)] \[\sum_{x\in \Z^d}\psi_{\alpha}(|x-y|)\ = \ \frac{C_d\alpha}{-E} \mbox{ for
any }y\in \Z^d\,; \]
\item[(c)] \[\psi_\alpha(|r\pm b|) \ \le\ C(\Theta)\,\psi_\alpha(r)\mbox{ for }
0<b<\diam(\Theta)\,;\]
\item[(d)]
\[\prod_{i=1}^{2n+1}\,\psi_\alpha(|x_{i-1}-x_i|) \ \le
e^{-|x_{2n+1}-x_0|}\frac{\sqrt{-E}}{2\alpha}\
\prod_{i=1}^{2n+1}\,\psi_{\alpha/2}(|x_{i-1}-x_i|)\,.\]
\end{enumerate}
Suppressing the subscript $E$ in the free Green function, we have
\[G(x,y)\ = \ \int_{\T^d}e^{i2\pi(x-y)p}\frac{d^dp}{e(p)-E}
\ =  \ \int_{\T^d}e^{i2\pi(x-y)p}\frac{d^dp}{e(p)-E}\,.\]  Let
$w=x-y$. For any given $w\in \Z^d$ let us choose
$\gamma\in\{1,\ldots,d\}$ so that
\begin{equation}\label{eq:gae'}
|w\cdot e_\gamma|=\max_{i\in\{1,\ldots,d\}} |w\cdot e_i|\,.
\end{equation}
Then
\begin{equation}\label{eq:bigcomp}
|w\cdot e_\gamma|\ge|w|/\sqrt{d}\,.
\end{equation}
Note that
\be\label{eq:decompmom}
\int dp\,\frac{1}{e(p)-E}\,e^{-i2\pi p\cdot w}\
= \ \int dq \,e^{-i2\pi q\cdot w}\int_{-1/2}^{1/2} d(p\cdot
e_\gamma) \,\frac{1}{e(p)-E}\,e^{-i2\pi (p\cdot e_\gamma  w\cdot
e_\gamma)}\,,
\ee
where $q$ stands for the $d-1$ dimensional vector obtained from $p$
by removing its $\gamma$ component (for $d=1$ the argument below becomes completely straightforward, so we will only consider $d\ge2$). Without loss of generality, let
us assume that $w\cdot e_\gamma>0$. Let 
\[\hat e( q)=2\sum_{\alpha\neq\gamma}\sin^2(\pi p\cdot e_\alpha)\,.\]It is easy to check that the
integrand as a function of $p\cdot e_\gamma$ is $1$-periodic,
analytic  inside the rectangle formed by the points
$$\{-1/2;\ -1/2+i\sqrt{\frac{\hat e(q)-E}{6d}};\ 1/2+i\sqrt{\frac{\hat e(q)-E}{6d}};\ 1/2\}$$ for a 
sufficiently small value of $-E>0$: Indeed, using 
$\sin(a+ib)=\sin a \cosh b +i \cos a \sinh b$ one can check that for any  $-1<E<0$ and $\epsilon$ satisfying 
\[0\le\epsilon\le
\sqrt{\frac{\hat e(q)-E}{6d}}\]  we have 
$${\rm Re}\, e(p+i\epsilon e_\gamma )-E\ge (e(p)-E)/2\,,$$
uniformly in $q$. 
 Moreover, the
periodicity implies that the integrals over the vertical segments
coincide:
\begin{multline}
\int_{-1/2}^{-1/2+i\sqrt{\frac{e(q)-E}{6d}}} d(p \cdot
e_\gamma)\frac{1}{e(p)-E}\,e^{-i2\pi (p\cdot e_\gamma  w\cdot
e_\gamma)}
\\ = \
\int_{1/2}^{1/2+i\sqrt{\frac{e(q)-E}{6d}}} d(p\cdot e_\gamma)
\,\frac{1}{e(p)-E}\,e^{-i2\pi (p\cdot e_\gamma  w\cdot
e_\gamma)}\,.
\end{multline}
Therefore
\begin{multline}
\left|\int_{-1/2}^{1/2} d(p\cdot e_\gamma)
\,\frac{1}{e(p)-E}\,e^{-i2\pi (p\cdot e_\gamma  w\cdot
e_\gamma)}\right| \\
= \
\left|\int_{-1/2+i\sqrt{\frac{e(q)-E}{6d}}}^{1/2+i\sqrt{\frac{e(q)-E}{6d}}}\
d(p\cdot e_\gamma) \,\frac{1}{e(p)-E}\,e^{-i2\pi (p \cdot e_\gamma
w\cdot e_\gamma)}\right|\\
\le \ 2\,e^{-w\cdot
e_\gamma\sqrt{\frac{e(q)-E}{6d}}}\int_{-1/2}^{1/2}
d(p\cdot e_\gamma) \,\frac{1}{e(p)-E}\\
\le \ 4\,e^{-|w|\frac{\sqrt{e(q)-E}}{3d}}\
\frac{1}{\sqrt{e(q)-E}}\,,
\end{multline}
where in the last step we used  \eqref{eq:bigcomp}. We can
consequently estimate the right hand side of \eqref{eq:decompmom} by
\[4\, \int dq \,e^{-|w|\frac{\sqrt{e(q)-E}}{3d}}\
\frac{1}{\sqrt{e(q)-E}}\,.\] To estimate the latter integral, we
split $\T^{d-1}$ into $B:=\{q\in \T^{d-1}: \ e(q)\le -E\}$ and
$\sim B:=\T^{d-1}\setminus B$. Then
\begin{multline}\label{eq:B}
\int_B dq \,e^{-|w|\frac{\sqrt{e(q)-E}}{3d}}\
\frac{1}{\sqrt{e(q)-E}} \ \le \ \int_B dq
\,e^{-|w|\frac{\sqrt{-E}}{3d}}\ \frac{1}{\sqrt{-E}} \\ \le \ C_d
\,e^{-|w|\frac{\sqrt{-E}}{3d}} \, (-E)^{(d-2)/2}\,,\end{multline}
and
\begin{multline}\label{eq:simB}
\int_{\sim B} dq \,e^{-|w|\frac{\sqrt{e(q)-E}}{3d}}\
\frac{1}{\sqrt{e(q)-E}} \ \le \ \int_{\sim B} dq
\,e^{-|w|\frac{\sqrt{e(q)}}{3d}}\ \frac{1}{\sqrt{e(q)}} \\ \le \
\int_{\sim B} dq \,e^{-|w|\frac{2}{3d}\sqrt{2q^2}} \,
\frac{\pi}{2\sqrt{2q^2}} \\ \le \   C_d
\,e^{-|w|\frac{2\sqrt{2-E}}{3d}}
\,\sum_{k=0}^{d-2}\frac{(d-2)!}{k!}\frac{(-E)^k}{|w|^{d-k-1}}
\,,\end{multline} for $d\ge2$, and where in the penultimate step we
have used Jordan's inequality. Summing up \eqref{eq:B} and
\eqref{eq:simB},  we arrive to \eqref{eq:freeGF}.

\end{proof}
Another useful property of the free Green function is captured by the following assertion:
\begin{lemma}\label{lem:freeGFprop}
For all $E<0$ and all $\Z^3\ni x\neq0$ we have
\begin{equation}\label{eq:freeGFbn}
\frac{1}{6-2E}\ < \ \frac{G_E(0,x)}{G_E(0,x+e)}\ <\ 6-2E\,,
\end{equation}
where $e\in\Z^3$ is any unit vector.
\end{lemma}
\begin{proof}
\hspace{7cm} \vskip 3mm \noindent
Note that for $x\neq 0$ one has
\[\langle 0| (-\frac{1}{2}\Delta-E)\,(-\frac{1}{2}\Delta-E)^{-1}|x\rangle \ = \ 0\,.\]
Inserting the partition of identity $I=\sum_{z\in\Z^3}|z\rangle\langle z|$ between the operators on the right hand side, and using \eqref{eq:Laplop}, we obtain the well known identity
\[\sum_{e\in\Z^3:\ |e|=1}G_E(-e,x) \ = \ (6-2E)G_E(0,x)\,.\]
By translation invariance, $G_E(-e,x)=G_E(0,x+e)$. Since $G_E(0,x)>0$ for any $x\in\Z^3$, we arrive to \eqref{eq:freeGFbn}.

\end{proof}
\section{Properties of the self energy $\sigma$}\label{sec:appendI}
\subsection{Properties of the solution of
\eqref{eq:self}}\label{subsec:sigmaprop}
In this section we establish the existence, periodicity, and
analyticity of the self energy operator $\sigma(p,E)$ introduced in \eqref{eq:self}.
We will use the following
inequalities (that can be deduced from \cite{Joyce}):
\be\label{eq:freegr} \int_{\T^3} \,d^3q\, \frac{1}{e(q)}< 2 \,,\quad
\int_{\T^3} \,d^3q\, \frac{1}{(e(q)+\epsilon^2)^2}<
\frac{1}{\epsilon}\,.\quad \ee
To prove the existence, we introduce the space
\[L(\T^3)\ =\
\{f\,:\ \T^3\rightarrow\C\,\big| \norm{f}_\infty < \infty\,,f\,\,
\mbox{ is real analytic}\}\,.\]
We define a map $T_\epsilon\,:\, L(\T^3)\rightarrow L(\T^3)$ pointwise as
\begin{equation}\label{eq:self1}
(T_\epsilon f)(p)=\lambda^2\int_{\T^3} \,d^3q\, \frac{\left|\hat
u(p-q)\right|^2}{e(q)-E-i\epsilon-f(q)}\,.
\end{equation}
We have $T_{\epsilon} B_{\beta}(0)\subset B_{\beta}(0)$, where
$B_{\beta}(0)$ is a ball (in $\|\cdot\|_\infty$ topology) of radius $\beta$ centered at the origin,
and
\[\beta \:= \ 2\lambda^2 \|\hat u\|_\infty^2\,.\]
Indeed, for $f\in B_{\beta}(0)$,
\[\lambda^2\left|\int_{\T^3} \,d^3q\, \frac{\left|\hat
u(p-q)\right|^2}{e(q)-E-i\epsilon-f(q)}\right| \le
\lambda^2\int_{\T^3} \,d^3q\, \frac{\|\hat u\|_\infty^2}{e(q)}<
2\lambda^2 \|\hat u\|_\infty^2\,.\]

Consider now the ball $B_\gamma(0)$ of the radius
\[\gamma\ :=\ \min(-E-2\lambda^4 \|\hat u\|_\infty^4,2\lambda^2 \|\hat u\|_\infty^2).\]
Then $T_\epsilon$ is a contraction on $B_\gamma(0)$. Indeed, let
$f,g\in B_\gamma(0)$, then
\begin{multline}\label{eq:fixpnt}
|(T_\epsilon f)(p)-(T_\epsilon g)(p)|\\ \le \lambda^2\int_{\T^3} \,d^3q\,
\frac{\left|\hat
u(p-q)\right|^2|f(q)-g(q)|}{|e(q)-E-i\epsilon-f(q)|\,|e(q)-E-i\epsilon-g(q)|}\\
\le \lambda^2\|f-g\|_\infty \int_{\T^3} \,d^3q\, \frac{C_\delta^2}
{(e(q)+2C^4\lambda^4)^2}\ \le
\frac{1}{\sqrt 2}\|f-g\|_\infty\,,
\end{multline}
where we have used \eqref{eq:freegr} in the last step. Hence by
the Banach fixed point theorem, the self consistent equation
\eqref{eq:self} has a single valued solution $\sigma(p,E+i\epsilon)$
for all $p\in \T^3$ and all
\[E<E_0:=-2\lambda^2 \|\hat u\|_\infty^2-2\lambda^4 \|\hat u\|_\infty^4\,.\]
The function $\sigma(p,E+i\epsilon)$ satisfies
\be\label{eq:bndonsigm}
\|\sigma\|_\infty\ \le \ \min(-E-2\lambda^4 \|\hat u\|_\infty^4,2\lambda^2 \|\hat u\|_\infty^2)\,.
\ee
\par Next we establish $1$-periodicity of the above solution (in the real space). To this end, we note that since $\hat{u}(p-q)=\sum_{n\in\Z^3}u(n)e^{2\pi i(p-q)n}$,
we have
\[\abs{\hat{u}(p-q)}^2=\sum_{m,n\in\Z^3}u(m)u(n)e^{2\pi i(p-q)(n-m)}\,.\]
Hence
\begin{multline}\label{eq:periodpr}
\lambda^2\int_{\T^3} \,d^3q\, \frac{\left|\hat
u(p-q)\right|^2}{e(q)-E-i\epsilon-f(q)} \\
=\lambda^2\sum_{m,n\in\Z^3} u(m)u(n)e^{2\pi ip(n-m)}
\int_{\T^3} \,d^3q\, \frac{e^{2\pi iq(m-n)}}{e(q)-E-i\epsilon-f(q)}\,,
\end{multline}
and periodicity of $\sigma(p,E+i\epsilon)$ follows from the periodicity of
$e^{2\pi ip(n-m)}$.
\par
Finally, we show analyticity. Fix a unit vector $e\in\Z^3$, and let $p_e:=p\cdot e$, $n_e:=n\cdot e$ for $n\in \Z^3$.  Using \eqref{eq:freegr}  one can readily check that
\[\left\|\frac{d^k\sigma(p,E+i\epsilon)}{d^k_{p_e}}\right\|_\infty \ \le \
2\lambda^2\sum_{m,n\in\Z^3}\abs{u(m)u(n)(2\pi)^k(n_e-m_e)^k}
\ \le \ C\,\lambda^2A^{-k+3}k!\,,
\]
for all $p_e\in\T$, and where $A$ is given by \eqref{eq:condonu}, and $C$ is some generic constant.
This implies that $\sigma(p,E+i\epsilon)$ is real analytic in $p_e$ variable and admits the complex analytic continuation to the rectangle ${\mathcal R }$ introduced in the paragraph followed by \eqref{eq:normr}. It follows from \eqref {eq:self} and \eqref{eq:freegr}  that we also have in this energy interval the bound
\be\label{eq:bndonsigmcomp}
\|\sigma\|_{\infty,{\mathcal R}}\ \le \ 2\lambda^2 \,\|\hat u\|_{\infty,{\mathcal R}} \,,
\ee
with the norm $\|\cdot\| _{\infty,{\mathcal R}}$   defined in
\eqref{eq:normr}.
\subsection{Properties of the solution of
\eqref{eq:self'}}
We proceed as in the previous subsection. We will be interested in the range of energies satisfying
\be\label{eq:rngEmat}E\ <  \ -\kappa\,,\ee
with $ \kappa=4n\,\lambda^2\,\|u\|^2_\infty $.
It follows from the block diagonal structure of the operator $\Sigma$ defined in
\eqref{eq:Sigm} that for any pair $\Sigma_1,\, \Sigma_2$ of such matrices
(which  correspond to $\sigma_1,\,\sigma_2$, accordingly) we have
$\|\Sigma_1-\Sigma_2\| = \|\sigma_1-\sigma_2\| $. Consider a ball
\[B:=\{\sigma\in M_{n,n}:\ \|\sigma\|\le\kappa/2\}\,,\]
 and a map
\[T:\ M_{n,n}\rightarrow M_{n,n}\,,\quad T\sigma\ :=\ \lambda^2\,D\,
S \, D \,,\]
where $S$ is defined in \eqref{eq:Selem} and $D$ in the paragraph followed by \eqref{eq:Sigm}. We claim that $TB\subset B$. Indeed, for any $\sigma\in B$, we have
\[\|T\sigma\| \ \le \ \lambda^2\,\|D\|^2\,\|S\|\,.\]
To estimate $\|S\|$, we observe that by construction of $S$,
\[\|S\| \ = \ \left\|P\,\left(-\Delta/2-E-i\epsilon-\Sigma\right)^{-1}\,P\right\|\,,\]
where $P$ is a projector onto $\supp\, \hat \Theta$. For the energies $E$ that satisfy \eqref{eq:rngEmat} and $\sigma \in B$ we have
\be\label{eq:KSig}K(\Sigma) \ := \ Re\, (-\Delta/2-E-i\epsilon-\Sigma)\ >\ -\Delta/2+ \kappa/2\,.\ee
But for an operator $K=A+iB$ with positive $A$ and hermitian $B$,
and a hermitian operator $F$ we have
\begin{multline*}\left\|F(A+iB)^{-1}F\right\| \ = \  \left\|FA^{-1/2}(I+iA^{-1/2}BA^{-1/2})^{-1}A^{-1/2}F\right\|\\ \le \  \left\|FA^{-1/2}\right\|\, \left\|(I+iA^{-1/2}BA^{-1/2})^{-1}\right\|\,\left\|A^{-1/2}F\right\| \ \le \   \left\|FA^{-1}F\right\|\,.\end{multline*}
Hence
\begin{multline*}\left\|P\,\left(-\Delta/2-E-i\epsilon-\Sigma\right)^{-1}\,P\right\|\ \le\ \left\|P\,\left(Re\,(-\Delta/2-E-i\epsilon-\Sigma\right))^{-1}\,P\right\| \\ \le \
\left\|P\,\left(-\Delta/2+\kappa/2\right)^{-1}\,P\right\| \ \le \ 2n\,,\end{multline*}
where in the last step we have used the fact that the norm of the
matrix is dominated by  its trace norm, the positivity of
$P\left(-\Delta/2+\kappa/2\right)^{-1}P$, and bound
\eqref{eq:freegr}. Putting everything together, we get
\[\|T\sigma\| \ \le \ 2n\,\lambda^2\,\|D\|^2 \ = \ \kappa/2\,.\]
By Brouwer's fixed point theorem the map $T$ then has  at least one  fixed point in
$B$. Since we are interested in proving the uniqueness, we will show that $T$ is also a
contraction on $B$. To this end, let $\sigma_{1,2}\in B$. Then using the second
resolvent identity
\begin{multline}\label{eq:contrsigm}
\|T \sigma_1-T\sigma_2\| \ \le \
\lambda^2\,\|D\|^2\,\|\Sigma_1-\Sigma_2\|\ \times\\
\left\|P\,\left(-\Delta/2-E-i\epsilon-\Sigma_1\right)^{-1}\right\|\,
\left\|\left(-\Delta/2-E-i\epsilon-\Sigma_2\right)^{-1}\,P\right\|\,.
\end{multline}
Now observe that
\begin{multline*} P\,\left(-\Delta/2-E-i\epsilon-\Sigma_1\right)^{-1}\,\left(-\Delta/2-E+i\epsilon-\Sigma^*_1\right)^{-1}\,P \\ = \ P\,K^{-1/2}(\Sigma_1)\left(I+iB\right)^{-1}\,K^{-1}(\Sigma_1)\,\left(I-iB\right)^{-1}K^{-1/2}(\Sigma_1)\,P\,,
\end{multline*}
where $B$ is a self adjoint operator. Using \eqref{eq:KSig}, we can bound the right hand side (in the operator sense) by

\[ \frac{2}{\mu}\,P\,K^{-1/2}(\Sigma_1)\left(I+iB\right)^{-1}\,\left(I-iB\right)^{-1}K^{-1/2}(\Sigma_1)\,P \, \le \, \frac{2}{\kappa}\,P\,K^{-1}(\Sigma_1)\,P \, \le \, \frac{4}{\kappa}\,,
\]
where in the last step we have used \eqref{eq:freegr}. As a result, we have obtained the bound
\[\left\|P\,\left(-\Delta/2-E-i\epsilon-\Sigma_1\right)^{-1}\right\|^2 \ \le\ \frac{4}{\kappa} \ = \ \frac{1}{n\lambda^2\|u\|^2_\infty} \,.\]
Using it and its analogue for $\Sigma_2$ in \eqref{eq:contrsigm}, we
arrive to
\[ \|T \sigma_1-T\sigma_2\| \ \le \ \frac{1}{n}\|\sigma_1-\sigma_2\|\ < \ \|\sigma_1-\sigma_2\|\,,\]
hence $T$ is a contraction on $B$. In summary, we have shown that \eqref{eq:self'} has a unique solution in the above energy interval, and
\be\label{eq:bndSi}
\|\Sigma\| \ = \ \|\sigma\| \ \le \ 2\,|\hat \Theta|\,\lambda^2\,\|D\|^2\,.
\ee
%
\subsection{Properties of $R_r$ in ($\mathcal N$) case}
Here we will consider the properties of the Green function $R_r(x,y)$ defined in \eqref{eq:Rsigma} where $\Sigma$ satisfies \eqref{eq:propsigm}. The following assertion holds:
\begin{lemma}\label{lem:propSigma}
Let $E_0$ be given by \eqref{eq:Esigm}. Then for $|\epsilon|<\kappa/2$, and all $E<E_0$ we have
\be \left|R_r(x,y)\right| \ \le \ \langle x|\,\left(-\Delta/2-E+E_0/2\right)^{-1}\,|y\rangle \,.\ee
\end{lemma}
\begin{proof}
\hspace{7cm} \vskip 3mm \noindent
We first expand $R_r$ in Neumann series
\[R_r \ = \ G\sum_{j=0}^\infty ((\Sigma+i\epsilon) G)^j\,,\]
with
\[G\ :=\ \left(-\Delta/2-E\right)^{-1}\,.\]
Since $\|\Sigma\| \le \kappa/2$, the series converges absolutely for $E<E_0$ and $|\epsilon|<\kappa/2$. To estimate each individual term in the expansion, we insert partitions of identity $I=\sum_{z\in\Z^3}|z\rangle\langle z|$ after each operator in the product. We obtain
\be \label{eq:terminexp}
\langle x|G(\Sigma G)^j|y\rangle \ = \ \sum_{\{z_k\}_{k=1}^{2j}} G(x,z_1)\,\prod_{l=1}^{j}\Sigma(z_{2l-1},z_{2l})\, G(z_{2l},z_{2l+1})\,,
\ee
with a convention $z_{2j+1}=y$. It follows from the construction of $\Sigma$ that
\[\Sigma(x,y) \ = \ 0 \hbox{ for } x-y\notin \hat \Theta\,.\]
Also, by  \eqref{eq:propsigm} we have
\[|\Sigma(x,y)| \ \le \ \kappa/2\,.\]
Using these bound together with the estimate \eqref{eq:freeGFbn} to estimate the left hand side of \eqref{eq:terminexp}, we get
\begin{multline} \label{eq:terminexp'}
\left|\langle x|G(\Sigma R)^j|y\rangle\right| \\ \le \ \left(\frac{\kappa}{2}\,\{(6-2E)^{\diam \hat\Theta}\,|\hat\Theta|+1\}\,\right)^j\,\sum_{\{z_k\}_{k=1}^{j}} G(x,z_1)\,\prod_{l=1}^{j}G(z_{2l},z_{2l+1})\\ = \  \left(\frac{-E_0}{2}\right)^{j}\,\langle x|G^{j+1}|y\rangle\,.
\end{multline}
Hence
\[\left|G_r(x,y)\right| \ \le \ \langle x|\,G\,\sum_{j=0}^\infty \left(-E_0/2\right)^{j}\,G^j\,| y \rangle = \ \langle x|\,G\,\left(I+E_0/2\cdot G\right)^{-1}\,| y \rangle\,.\]
But
\be G\,\left(I+E_0/2\cdot G\right)^{-1} \ = \ \left(-\Delta/2-E+E_0/2\right)^{-1} \,,\ee
hence the result follows.

\end{proof}
\subsection{Dipole single site potential}
Here we consider a special case of the single site potential $u_d$ defined in \eqref{eq:defdip}
Let  $T_\epsilon$ be the same map as the one defined in \eqref{eq:self1}. Then $|\hat u (p)|^2=4\sin^2(\pi p\cdot e_1)$,  and for the even function $f$ we have
\begin{multline}\label{eq:T_d}
(T_\epsilon f)(p)\ = \ 4\lambda^2\sin^2(\pi p\cdot e_1)\int_{\T^3} \,d^3q\, \frac{\cos(2\pi q\cdot e_1)}{e(q)-E-i\epsilon-f(q)}\\+\,4\lambda^2\int_{\T^3} \,d^3q\, \frac{\sin^2(\pi q\cdot e_1)}{e(q)-E-i\epsilon-f(q)}\,.
\end{multline}
 We will consider the energies $E$  that satisfy
\[ E\ < \ -\,(1+\lambda)\lambda^2 \,.\]
   The subspace $G$ of $ L^\infty(\T^3)$ consisting of the functions $f(p)=A + B\sin^2(\pi p\cdot e_1)$ is clearly invariant under the map  $T_\epsilon$. Let $G'$ denote an open subset of $G'$ characterized by
$|A|<\lambda^2$, $|B|<14\lambda^2$. We then have $T_\epsilon G'\subset G'$. Indeed, for
$f\in G'$, we can estimate the two terms on the right hand side of \eqref{eq:T_d} using two bounds below, that hold for $\lambda$ sufficiently small:
\begin{multline}\label{eq:T_d1}
\int_{\T^3} \,d^3q\, \frac{|\cos(2\pi q\cdot e_1)|}{|e(q)-E-i\epsilon-f(q)|}\ < \
\int_{\T^3} \,d^3q\, \frac{1+2\sin^2(\pi p\cdot e_1)}{(1-|B|)e(q)}\\ = \ \int_{\T^3} \,d^3q\, \frac{1}{(1-|B|)e(q)}\,+\,\frac{1}{3(1-|B|)} \ < \ \frac{7}{2}\,,
\end{multline}
where in the penultimate step we have used the symmetry of the integral with respect to spatial directions $\{1,2,3\}$ and in the last step we have used \eqref{eq:freegr}. The second estimate we need is
\be\label{eq:T_d2}
\int_{\T^3} \,d^3q\, \frac{2\sin^2(\pi q\cdot e_1)}{|e(q)-E-i\epsilon-f(q)|} \ < \ \int_{\T^3} \,d^3q\, \frac{2\sin^2(\pi q\cdot e_1)}{(1-|B|)e(q)} \  = \ \frac{1}{3(1-|B|)} \ < \ \frac{1}{2}\,.
\ee
Combining these two bounds we obtain
\be
\left|(T_\epsilon f)(p)\right|\ < \ 14\lambda^2\sin^2(\pi p\cdot e_1)\,+\,\lambda^2\,,
\ee
hence $T_\epsilon f\in G'$. Since $G'$ is a compact convex set, one can use Brouwer's fixed point theorem to conclude the existence of the fixed point (in fact one can use this technique to show the existence of the fixed point for {\it all} negative values of $E$). However, we also need a uniqueness of the fixed point, so we proceed to prove that $T_\epsilon$ is a contraction on $G'$. To this end, let us introduce a norm on $G$:
\[\|f\|_G \ =\ |A|\,+\,\lambda |B|\,,\quad \hbox{ for } f=A + B\sin^2(\pi p\cdot e_1)\,.\]
Let
$f=A + B\sin^2(\pi p\cdot e_1),g=C + D\sin^2(\pi p\cdot e_1)\in G'$, then the straightforward computation similar to the one done in \eqref{eq:fixpnt} gives
\begin{multline} \|(T_\epsilon f)(p)-(T_\epsilon g)(p)\|_G\\
\le \ 4\lambda^2 \int_{\T^3} \,d^3q\, \frac{(|B-D|+|A-C|)\sin^2(\pi q\cdot e_1)}
{\left((1-|B|)e(q)+\lambda^3\right)\left((1-|D|)e(q)+\lambda^3\right)}\\ +\,
4\lambda^3 \int_{\T^3} \,d^3q\, \frac{|A-C|+|B-D|\sin^2(\pi q\cdot e_1)}
{\left((1-|B|)e(q)+\lambda^3\right)\left((1-|D|)e(q)+\lambda^3\right)}
 \\ < \ 20\lambda^2|B-D|\,+\,5\lambda^{3/2}|A-C| \ < \
20\lambda\|f-g\|_G\,,
\end{multline}
for $\lambda$ small enough. We have used \eqref{eq:freegr} in the penultimate step. Hence by
Banach fixed point theorem, the self consistent equation
\eqref{eq:self} has a single valued solution $\sigma(p,E+i\epsilon)$
for all $p\in \T^3$ and all
\[E\ < \  E_d\ :=\ -\,(1+\lambda)\lambda^2\,.\]
Since $\sigma\in G'$ we have
\be\label{eq:sigma_d}
\sigma(p,E+i\epsilon) \ = \ A + B\sin^2(\pi p\cdot e_1)\,;\quad |A|\ <\ \lambda^2\,,\ |B|<\ 14\ \lambda^2\,.
\ee
It follows from the functional form of  $\sigma(p,E+i\epsilon)$ that for any unit vector $e\in\Z^3$  the function $\sigma$ is $1$-periodic,  analytic in
$p_e:=p\cdot e$ (in fact it is a constant unless $e=e_1$).  Let $E_d^*$ be a parameter that satisfies $0<E_d^*< E_d\,-\, E$ and let
\[\delta \ :=\ \sqrt{(E_d-E-E_d^*)/2}\,.\]
Then using \eqref{eq:sinext} and \eqref{eq:sigma_d} we deduce that for an  arbitrary  $p\in\T^3$  we have
\be\label{eq:sigma_dcomp}
Re\left(e(p+i\delta e)-E-i\epsilon-\sigma(p+i\delta e,E+i\epsilon)\right)\ >\  (1-5\lambda^2)\left(e(p)+E_d^*\right) \,.
\ee
%
\begin{proof}[Proof of Proposition \ref{prop:dipole}]
\hspace{7cm} \vskip 3mm \noindent
Let $\Lambda:[-L,L]^3\cap \Z^3$, $\Lambda_+:[-L-1,L+1]\times[-L,L]^2\cap \Z^3$  and let $\Omega_\Lambda:=\times_{k \in \Lambda_+} \R$. By the standard arguments (c.f. the discussion in Section 6 of \cite{Kirsch-89a}) it suffices to find a configuration $\omega\in\Omega_\Lambda$, for which  $\min\sigma(H^\Lambda_\omega)<-2\lambda^2+O(\lambda^3)+o(1)$. Here $o(1)$ is taken with respect to the $L$ variable. We choose $\omega$ in such a way that for $x\in\Lambda$, $V^\Lambda_\omega(x)=-2\lambda$ for $x\cdot e_1=0$ and $V^\Lambda_\omega(x)=0$ otherwise. Clearly, the bottom of the spectrum of $H^\Lambda_\omega$ converges, in the limit $L\rightarrow\infty$, to $\inf \sigma(\hat H)$, where the latter operator acts on the whole $\Z^3$ as
\[\hat H \ = \ -\frac{\Delta}{2}\,+\,\hat V\,,\]
with
\[\hat V(x)=-2\lambda  \hbox{ for }x\cdot e_1=0 \hbox{ and }\hat V(x)=0 \hbox{ otherwise}\,.\]
Readily, $\inf \sigma(\hat H)\le\min\sigma( \tilde H)$, where $\tilde H$ is a one dimensional restriction of $\hat H$ to the $e_1$ direction. However, $\tilde H$ is a rank one perturbation of the free Laplacian. It  follows from the rank one perturbation theory that $E_m:=\min\sigma (\tilde H)$ is given by the solution of the equation
\[2\lambda \ = \ G_{00}(E_m)\,,\]
where $G$ is a free one dimensional Green function. Using the Fourier transform, the above equation can be rewritten as
\[\frac{1}{2\lambda} \ = \ \int_{\T} \,\frac{dq}{2\sin^2(\pi q)-E_m}\,.\]
Finally, since
\[\int_{\T} \,\frac{dq}{2\sin^2(\pi q)-E_m} \ = \ \int_{-\infty}^\infty \,\frac{dq}{2\pi^2 q^2-E_m}\,+\,O(1) \ = \ \frac{1}{\sqrt{-2 E_m}}\,+\,O(1)\]
which holds for $E_m<0$, we obtain
\[E_m \ = \ -2\lambda^2  +O(\lambda^3)\,.\]
The rest of the argument coincides with the one of Theorem \ref{thm:main} for the ($\mathcal O$) case. One just need to replace the subscript $\mathcal O$ with $d$ everywhere in the proof.

\end{proof}

\paragraph{Acknowledgements}
It is a pleasure to thank Professor G\"unter Stolz for several
useful discussions concerning the dipole model.  



\end{document}